\documentclass{article}
\usepackage[english]{babel}
\usepackage{amsthm, amsmath, booktabs, graphicx, amssymb, float, wrapfig, tikz, pdfpages, parskip, blindtext, hyperref, siunitx, subcaption, listings, xcolor, cite}
\newtheorem{theorem}{Theorem}[section]
\newtheorem{lemma}[theorem]{Lemma}
\newtheorem{remark}[theorem]{Remark}

\usepackage[title]{appendix}
\usepackage{makecell}
\usepackage[ruled,vlined]{algorithm2e}
\usepackage[a4paper, margin=1in]{geometry}
\usepackage{comment}
\title{Bayesian Model Selection with an Application to Cosmology}
\author{Nikoloz Gigiberia}
\date{}
\begin{document}

\maketitle

\begin{abstract}
We investigate cosmological parameter inference and model selection from a Bayesian perspective. Type Ia supernova data from the Dark Energy Survey (DES-SN5YR) are used to test the \(\Lambda\)CDM, \(w\)CDM, and CPL cosmological models. Posterior inference is performed via Hamiltonian Monte Carlo using the No-U-Turn Sampler (NUTS) implemented in NumPyro and analyzed with ArviZ in Python. Bayesian model comparison is conducted through Bayes factors computed using the \texttt{bridgesampling} library in R. The results indicate that all three models demonstrate similar predictive performance, but \(w\)CDM shows stronger evidence relative to \(\Lambda\)CDM and CPL. We conclude that, under the assumptions and data used in this study, \(w\)CDM provides a better description of cosmological expansion.
\end{abstract}

\section{Introduction}
\label{sec:intro}

In modern cosmology, Type Ia supernovae (SNe Ia) have proven to be among the most informative cosmological probes, providing strong evidence for the accelerated expansion of the Universe. By serving as standardisable candles, SNe Ia enable precise measurements of the luminosity distance--redshift relation, which constrains both the matter density and the properties of dark energy. However, despite major observational and theoretical progress, the fundamental origin of cosmic acceleration and the nature of dark energy remain unresolved.  Bayesian statistical methods offer a coherent framework for addressing these questions. In addition to enabling robust parameter estimation, Bayesian inference provides a principled basis for comparing competing cosmological models based on data. 

In this paper, we apply Bayesian inference through Hamiltonian Monte Carlo (HMC) to the DES-SN5YR supernova dataset \cite{DES_SN5YR_data} under three dark energy scenarios: the standard $\Lambda$CDM model, its extension with a constant equation of state parameter (\(w\)CDM), and the Chevallier-Polarski-Linder (CPL) parametrization with time-varying dark energy. Using posterior sampling and bridge sampling for model comparison, we demonstrate how Bayesian tools can be used to constrain cosmological parameters and assess the relative evidence for competing descriptions of dark energy.  

\section{Cosmological Theory}
\label{sec:cosmo_theory}

The expansion of the universe is governed by Einstein's theory of General Relativity, whose dynamics on large scales are described by the Friedmann equations. These equations are formed under the assumption that the universe is homogeneous and isotropic. Under this assumption, we are led to the Friedmann–Lemaître–Robertson–Walker (FLRW) metric [For more context on the assumptions on the models present in this paper, refer to \hyperref[sec:model_assump]{Section 2.1}]. This metric provides a mathematical framework for describing a dynamic spacetime and is given by:
\begin{equation}
ds^2 = -c^2 dt^2 + a(t)^2 \left[ \frac{dr^2}{1 - k r^2} + r^2 \left( d\theta^2 + \sin^2\theta \, d\phi^2 \right) \right],
\end{equation}
where \( ds^2 \) is the spacetime interval, \( c \) is the speed of light, \( a(t) \) is the scale factor describing how distances evolve with cosmic time \( t \), \( k \) is the spatial curvature constant (with values $0$, $+1$, or $-1$ for flat, closed, or open geometries respectively), and \( (r, \theta, \phi) \) are co-moving spatial coordinates. The FLRW metric serves as the backbone of the standard model of cosmology, allowing the Friedmann equations to relate the expansion rate to the energy content of the universe.

The Friedmann equations describe the evolution of the expansion rate, quantified by the Hubble parameter \(H(z)\), as a function of redshift \(z\). Redshift is defined observationally through the stretching of photon wavelengths,
\[
1 + z = \frac{\lambda_{\mathrm{obs}}}{\lambda_{\mathrm{em}}}.
\]
In an FLRW universe, where distances scale with the scale factor \(a(t)\), one can show that photon wavelengths stretch proportionally to \(a(t)\). Under this assumption, the observational definition of redshift becomes
\[
1 + z = \frac{a_0}{a(t_{\mathrm{em}})},
\]
where \(a_0\) is the present-day scale factor and \(a(t_{\mathrm{em}})\) is the value at the time the light was emitted. A higher redshift corresponds to a more distant and earlier epoch in the universe’s history. In the context of Type Ia supernovae, the observed redshift allows us to reconstruct the expansion history of the universe by comparing the predicted and observed luminosity distances.

Dark energy is a hypothetical form of energy with negative pressure, responsible for the observed acceleration of the universe's expansion. It is characterized through its equation-of-state parameter \( w \), defined as:
\begin{equation}
w = -\frac{p}{\rho},
\end{equation}
where \( p \) is the pressure and \( \rho \) is the energy density of the dark energy component. The value of \( w \) determines how dark energy evolves with time. If \( w = -1 \), this implies a constant energy density over time, while deviations from this value allow for dynamical models of dark energy such as quintessence or more general parametric forms like the CPL parametrization \cite{Ryden2016}.

\subsection{Model Assumptions}
\label{sec:model_assump}

In this paper, we adopt several assumptions supported by results in modern cosmology to simplify the inference of cosmological parameters. First, we assume a spatially flat Universe, i.e., \(\Omega_k = 0\), consistent with the observations from the Planck 2018 release, which report \(\Omega_k = 0.0007 \pm 0.0019\) \cite{Planck2018}, as well as other contemporary cosmological analyses including the DES collaboration analysis \cite{DES2024}. This assumption allows us to simplify our cosmological models without significantly impacting the results. We further assume that the Universe is homogeneous and isotropic on large scales, which is well supported by observations of the Cosmic Microwave Background (CMB) and large-scale structure surveys. This leads naturally to the use of the Friedmann-Lemaître-Robertson-Walker (FLRW) metric as the spacetime background \cite{Ryden2016}.

We also assume that for the low-redshift supernova data (\(z \lesssim 1.5\)), the contribution of radiation to the energy budget becomes negligible compared to matter and dark energy. The radiation density parameter \(\Omega_r\) scales as \((1+z)^4\), and at \(z = 0\), its contribution is on the order of \(10^{-5}\), as inferred from CMB temperature and effective neutrino species. We therefore neglect \(\Omega_r\) in the models in this paper.

Finally, we treat the Hubble constant \(H_0\) as a free parameter rather than fixing it to values from past studies such as those of Planck or SH0ES, given the significant tension between local and early-Universe measurements \cite{Riess2022}. 

\subsection{Cosmological Models}
\label{sec:cosmo_models}

Among all three cosmological models, the Hubble parameter, which will be at the core of the model evaluation process in this paper, has the general formula:
\begin{equation}
H(z)^2 \;=\; H_0^2\; \left( \Omega_m (1+z)^3 +\; \Omega_\mathrm{\Lambda} \exp\!\left[ 3 \int_0^z \frac{1+w(z')}{1+z'} \, dz' \right] \right)\; ,
\end{equation}
\[
\text{with }\qquad \Omega_\mathrm{\Lambda}=1-\Omega_m.
\]
given the assumptions in the previous section \cite{Ryden2016}. Given the relationship between \(\Omega_{\Lambda}\) and \(\Omega_m\), henceforth we will only be using \(\Omega_m\) both in inference and formulae.

\subsection{\texorpdfstring{ \( \Lambda CDM\) Model}{LambdaCDM Model}}

The $\Lambda$CDM (Lambda Cold Dark Matter) model is the standard model of cosmology. It assumes a spatially flat universe dominated by cold dark matter (CDM) and a cosmological constant $\Lambda$, which serves as a form of dark energy with constant energy density and equation of state $w = -1$. In this model, dark energy does not evolve with time \cite{Ryden2016}. Under the $\Lambda$CDM framework, the Hubble parameter $H(z)$, which quantifies the expansion rate at redshift $z$, is given by:
\begin{equation}
H(z)^2 = H_0^2 \left[ \Omega_m (1+z)^3 + 1-\Omega_m \right].
\end{equation}

\subsection{\texorpdfstring{ \( wCDM \) Model}{wCDM Model}}

The $w$CDM model generalizes $\Lambda$CDM by replacing the cosmological constant with a dark energy component whose equation of state $w$ is allowed to differ from $-1$, but is constant in time. The Hubble parameter  in this model, after solving the relevant integral, becomes:

\begin{equation}
H(z)^2 = H_0^2 \left[ \Omega_m (1+z)^3 + (1 - \Omega_m)(1+z)^{3(1+w)} \right],
\end{equation}

Here, $w$ is treated as a free parameter to be inferred from data. If $w = -1$, the model reduces to $\Lambda$CDM. Values of $w < -1$ correspond to ``phantom``  dark energy, while $w > -1$ imply a less rapidly accelerating universe \cite{An2016}.

\subsection{\texorpdfstring{ \( CPL \) Model}{CPL Model}}

The Chevallier-Polarski-Linder (CPL) model \cite{Chevallier2001, Linder2003} further generalizes dark energy by allowing its equation of state to evolve with redshift. This model introduces two parameters, \( w_0 \) and \( w_a \), and defines the equation of state as:
\[
w(z) = w_0 + w_a \frac{z}{1 + z}
\]
This parametrization captures a wide range of possible dark energy behaviors with a minimal number of parameters, reducing to \( w = w_0 \) at present (\( z = 0 \)) and allowing for linear evolution at higher redshifts.

The corresponding Hubble parameter is:
\begin{equation}
H(z)^2 = H_0^2 \left[ \Omega_m (1+z)^3 + (1 - \Omega_m)(1+z)^{3(1+w_0 + w_a)} \exp\left(-\frac{3 w_a z}{1+z} \right) \right]
\end{equation}
See the appendix for the full derivation.

\subsection{The Distance Modulus}
\label{sec:dist_mod}

Observationally, the DES dataset, and Type Ia supernova surveys generally, report measurements in terms of distance modulus \(\mu(z)\), allowing theoretical predictions of \(\mu(z)\) as a function of \(z\) to be compared directly with the data. The distance modulus can be calculated theoretically through the luminosity distance \(d_L(z)\), which encodes the expansion history of the universe. For a flat universe, the luminosity distance is given by  
\[
d_L(z) = (1+z) \, c \int_0^z \frac{dz'}{H(z')}
\]  
where \(c\) is the speed of light, and \(H(z)\) is the Hubble parameter at redshift \(z\) \cite{Ryden2016}. Once \(d_L(z)\) is computed, the distance modulus is  
\begin{equation}
\mu(z) = 5 \log_{10} \left( \frac{d_L(z)}{1\, \text{Mpc}} \right) + 25
\end{equation}  
with \(d_L(z)\) expressed in megaparsecs (Mpc). This formula allows us to utilise redshift data to make theoretical predictions of \(\mu\), allowing for Bayesian inference by comparing observed supernova distance moduli to model predictions as a function of cosmological parameters.

\section{Data}
\label{sec:data}

In this paper, we use the 5-Year Supernova sample from the Dark Energy Survey (DES-SN5YR) \cite{DES_SN5YR_data}, which includes light curves and associated data products for 1829 spectroscopically and photometrically classified Type Ia supernovae. The light curves are obtained using the Scene Modeling Photometry (SMP) pipeline, as detailed in \cite{Brout2019} and \cite{Sanchez2024}. 

The DES-SN5YR supernova dataset provides observed distance moduli  \(\mu_{obs}\),  their associated measurement uncertainties \(\sigma_{\mu}\) and a covariance matrix containing correlated systematic uncertainties. While the statistical errors are specific to each supernova and treated as independent, the systematic uncertainties arise from shared sources, such as calibration errors, which induce correlations across multiple supernovae. These correlations are encoded in the aforementioned covariance matrix provided by the DES collaboration \cite{DES2024, Vincenzi2024}.

To construct the total uncertainty matrix \(\Sigma\) we begin by building a diagonal matrix from the squared statistical uncertainties \(\sigma_{\mu}^2\), assuming that these errors are uncorrelated and independent between supernovae. We then add the systematic covariance matrix to this diagonal statistical matrix. This yields the full covariance matrix, \(\Sigma\ = \Sigma_{stat} + \Sigma_{sys}\), which accounts for both uncorrelated and correlated uncertainties. This matrix serves as the noise model in a multivariate Gaussian likelihood of the form \(\mu_{obs} \sim \mathcal{N}(\mu_{theory}, \Sigma)\) \cite{DES_SN5YR_data}, where \(\mu_{theory}\) is the predicted distance modulus computed from a given cosmological model. 

The figure below presents the Hubble diagram for the DES-SN5YR dataset. The Hubble diagram shows the observed distance modulus of Type Ia supernovae as a function of redshift. Each point corresponds to a supernova measurement, with the vertical axis indicating relative luminosity distance and the horizontal axis its cosmological redshift. The three theoretical models which will be analysed in this paper predict a specific curve in this diagram given the data. The aim of the analysis is to fit such a curve to the data, allowing one to estimate cosmological parameters and assess which models of dark energy best describe the observed relation through the Bayesian paradigm.  

\begin{figure}[H]
    \centering
    \includegraphics[width=1\linewidth]{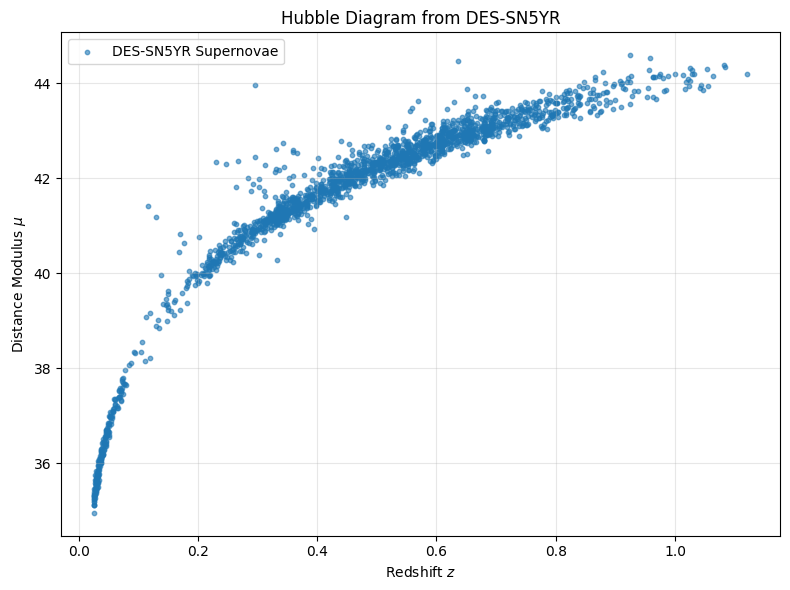}
    \caption{The Hubble Diagram for the DES-SN5YR dataset.}
    \label{fig:Hubble-Diagram}
\end{figure}

\section{Models and Methodology}
\label{sec:models_methods}

\subsection{Bayesian Setup}
\label{sec:Bayesian}

Bayesian inference provides a framework for parameter estimation and model selection/comparison which combines prior beliefs and data. Suppose we have a model \(M\) with parameters \(\theta\), and some observed data \(D\). Using Bayes' theorem we can find that:
\begin{equation}
    f(\theta|D, M) = \frac{f(D|M, \theta)f(\theta|M)}{f(D|M)}
\end{equation}
where \( f(\theta|D, M)\) is the posterior distribution, \(f(\theta|M)\) is the prior distribution, \(f(D|M, \theta)\) is the likelihood, and \(f(D|M)\) is the marginal likelihood or evidence, which is obtained by integrating:

\begin{equation}
    f(D|M) = \int f(D|M, \theta)f(\theta|M) d\theta.
\end{equation}
This integral averages the likelihood function $f(D |M, \theta)$ across parameter values $\theta$, weighted by their prior probabilities $f(\theta | M)$.

To perform Bayesian model selection, we then use the above reasoning to update our prior beliefs and perform an inference on the cosmological parameters within the \(\Lambda CDM\), \(w CDM\), and \(CPL\) models. However, in most real world cases such as inferring cosmological parameters, the evidence and posterior are difficult to compute analytically. We thus opt to use  methods to approximate the posterior. To do this we make use of Markov Chain Monte Carlo (MCMC) methods in order to approximate the posterior distribution of each parameter by constructing a Markov chain whose equilibrium distribution is the target posterior. This is possible because the Markov chain is constructed with transition dynamics that leave the posterior distribution invariant. In particular, algorithms such as Metropolis–Hastings enforce detailed balance with respect to the posterior, ensuring that it is the stationary distribution of the chain. Under irreducibility and aperiodicity, the chain is ergodic and converges to this equilibrium, so that its samples asymptotically approximate the posterior. Traditional MCMC techniques such as the Metropolis-Hastings algorithm often suffer from inefficiencies in high-dimensional or correlated parameter spaces \cite{Robert2007}, making them suboptimal for cosmological models where degeneracies exist, such as that between \(\Omega_m\) and \(H_0\) and the dark energy equation of state parameters in the CPL model, \(w_0\) and \(w_a\). In order to reduce the chance of encountering such potential issues, we opt to use Hamiltonian Monte Carlo (HMC).

\subsection{Hamiltonian Monte Carlo}
\label{sec:HMC}

Hamiltonian Monte Carlo borrows the Hamiltonian dynamic concept from physics to sample from  a distribution. Hamiltonian dynamics describes the evolution of  a physical system over time using positions and momenta governed by Hamilton’s equations. In the context of probabilistic modeling, the HMC method introduces a probabilistic system analogous to a physical one \cite{neal2011mcmc, betancourt2017conceptual}. 
 
The method defines the system using two vectors; the model parameters we wish to sample \( \left( \boldsymbol\theta = (\theta_1, \theta_2, \ldots, \theta_d) \right) \)  act as the position vector, and a vector of auxiliary variables \( \left( \boldsymbol\phi = (\phi_1, \phi_2, \ldots, \phi_d) \right) \), acting as the momentum vector. 

Each point in the system's orbit is associated with a potential energy,  \(U(\boldsymbol{\theta})\), and a kinetic energy, \(K(\boldsymbol{\phi})\). The Hamiltonian, denoted by \begin{equation}
H(\boldsymbol{\theta}, \boldsymbol{\phi}) = U(\boldsymbol{\theta}) + K(\boldsymbol{\phi})
\end{equation}
The temporal evolution of the system is then governed by Hamilton's equations:
\begin{align}
\frac{\partial \theta_a}{\partial t} &=  \frac{\partial H}{\partial \phi_a} =\frac{\partial K(\boldsymbol{\phi})}{\partial \phi_a} \\
\frac{\partial \phi_a}{\partial t} &=  -\frac{\partial H}{\partial \theta_a} = -\frac{\partial U(\boldsymbol{\theta})}{\partial \theta_a}, \quad a =1, \ldots, d
\end{align}
The equations (2) and (3) describe how position and momentum change of a pseudo time variable \(t\), and allow us to simulate the trajectory of a particle through the parameter space.

In order to numerically solve (2) and (3), and thus apply it to HMC, we use the leapfrog integrator \cite{hoffman2014nuts}.  It approximates the trajectory by alternating updates of momentum and position by the following algorithm:

\begin{algorithm}[H]
\caption{Leapfrog Integrator}
\KwIn{Initial position $\boldsymbol{\theta}$, momentum $\boldsymbol{\phi}$, step size $\epsilon$, number of steps $L$}
\KwOut{Updated $(\boldsymbol{\theta}, \boldsymbol{\phi})$}

$\boldsymbol{\phi} \gets \boldsymbol{\phi} - \frac{\epsilon}{2} \nabla_{\boldsymbol{\theta}} U(\boldsymbol{\theta})$ 

\For{$i \gets 1$ \KwTo $L$}{
    $\boldsymbol{\theta} \gets \boldsymbol{\theta} + \epsilon \nabla_{\boldsymbol{\phi}} K(\boldsymbol{\phi})$
    \If{$i \ne L$}{
        $\boldsymbol{\phi} \gets \boldsymbol{\phi} - \epsilon \nabla_{\boldsymbol{\theta}} U(\boldsymbol{\theta})$ 
    }
}

$\boldsymbol{\phi} \gets \boldsymbol{\phi} - \frac{\epsilon}{2} \nabla_{\boldsymbol{\theta}} U(\boldsymbol{\theta})$ 

\KwRet{$(\boldsymbol{\theta}, \boldsymbol{\phi})$}
\end{algorithm}

Repeating these steps for \(L\) iterations simulates the evolution of the system for a total "simulation time" of \(L\epsilon\). 

In statistical mechanics, the canonical distribution describes the probability of a system being in a particular state with energy \(E(x)\):
\begin{equation}
    \psi(x) \propto e^{-E(x)} 
\end{equation}
In the HMC analogy, we associate this idea with the joint distribution over \((\boldsymbol{\theta}, \boldsymbol{\phi})\):
\begin{equation}
    \psi(\boldsymbol{\theta}, \boldsymbol{\phi}) \propto e^{-H(\boldsymbol{\theta}, \boldsymbol{\phi})}  = e^{-U(\boldsymbol{\theta})}e^{-K( \boldsymbol{\phi})}
\end{equation}
This factorization implies that the momentum and parameter distributions are independent. In this case, \(U(\boldsymbol{\theta}) = - \log f(\boldsymbol{\theta} | D)\) which is the log-posterior,  and \(K(\boldsymbol{\phi}) = \frac{1}{2}  \boldsymbol{\phi}^T \boldsymbol{C}^{-1}\boldsymbol{\phi}\) given that \(\boldsymbol{\phi} \sim \mathcal{N}(0, \boldsymbol{C}) \).

Thus, the Hamiltonian equations become:
\begin{align}
    \frac{\partial \theta_a}{\partial t} &=  [\boldsymbol{C}^{-1}\boldsymbol{\phi}]_a \\
\frac{\partial \phi_a}{\partial t} &=  \frac{\partial \log f(\boldsymbol{\theta} | D)}{\partial \theta_a} , \quad a =1, \ldots, d
\end{align}

Now we are ready to move on to the HMC algorithm itself. Rather than sampling from a proposal distribution, HMC simulates dynamics using the leapfrog integrator to propose future states that are likely to be accepted.  Hamiltonian Monte Carlo produces proposals by simulating approximate Hamiltonian dynamics. These dynamics are volume-preserving (symplectic) and reversible, which means that when combined with a Metropolis correction step, the Markov chain has the desired posterior as its invariant distribution \cite{neal2011mcmc, betancourt2017conceptual}. 

The HMC algorithm proceeds as follows:

\begin{algorithm}[H]
\caption{Hamiltonian Monte Carlo (HMC)}
\KwIn{Target density $\pi(\boldsymbol{\theta}) \propto \exp(-U(\boldsymbol{\theta}))$;\\
\hspace{1.5em} Mass matrix $\mathbf{C}$; Step size $\epsilon$; Leapfrog steps $L$; Number of samples $N$}
\KwOut{Samples $\{\boldsymbol{\theta}_1, \dots, \boldsymbol{\theta}_N\}$}

Initialize $\boldsymbol{\theta}_0$\;

\For{$n = 1$ \KwTo $N$}{
    Draw $\boldsymbol{\phi}_n \sim \mathcal{N}(\mathbf{0}, \mathbf{C})$\;
    Set $(\boldsymbol{\theta}, \boldsymbol{\phi}) \gets (\boldsymbol{\theta}_{n-1}, \boldsymbol{\phi}_n)$\;

    $\boldsymbol{\phi} \gets \boldsymbol{\phi} - \dfrac{\epsilon}{2} \nabla_{\boldsymbol{\theta}} U(\boldsymbol{\theta})$\;
    \For{$i = 1$ \KwTo $L$}{
        $\boldsymbol{\theta} \gets \boldsymbol{\theta} + \epsilon\, \mathbf{C}^{-1} \boldsymbol{\phi}$\;
        \If{$i < L$}{
            $\boldsymbol{\phi} \gets \boldsymbol{\phi} - \epsilon\, \nabla_{\boldsymbol{\theta}} U(\boldsymbol{\theta})$\;
        }
    }
    $\boldsymbol{\phi} \gets \boldsymbol{\phi} - \dfrac{\epsilon}{2} \nabla_{\boldsymbol{\theta}} U(\boldsymbol{\theta})$\;

    $\boldsymbol{\phi} \gets -\boldsymbol{\phi}$\;

    $H_{\text{start}} \gets U(\boldsymbol{\theta}_{n-1}) + \dfrac{1}{2} \boldsymbol{\phi}_n^\top \mathbf{C}^{-1} \boldsymbol{\phi}_n$\;
    $H_{\text{proposal}} \gets U(\boldsymbol{\theta}) + \dfrac{1}{2} \boldsymbol{\phi}^\top \mathbf{C}^{-1} \boldsymbol{\phi}$\;

    $\rho \gets \min\left(1, \exp\left(H_{\text{start}} - H_{\text{proposal}}\right)\right)$\;

    Draw $u \sim \text{Uniform}(0,1)$\;

    \eIf{$u < \rho$}{
        Accept: $\boldsymbol{\theta}_n \gets \boldsymbol{\theta}$\;
    }{
        Reject: $\boldsymbol{\theta}_n \gets \boldsymbol{\theta}_{n-1}$\;
    }
}
\end{algorithm}

\subsection{The NUTS Sampler}
\label{sec:NUTS}

While HMC is efficient in exploring complex posterior landscapes, it requires careful tuning of the step size $\epsilon$ and the number of steps \(L\), and improper choices can lead to poor mixing or high rejection rates.

The No-U-Turn Sampler (NUTS) \cite{hoffman2014nuts} eliminates the need to set the trajectory length $L$ by adaptively determining how long to simulate the Hamiltonian dynamics. The core idea is to stop the simulation once the trajectory begins to retrace its path, the so-called U-turn.

NUTS constructs a balanced binary tree of candidate points by simulating leapfrog steps both forward and backward in time. At each doubling of the tree, it checks a no U-turn condition that terminates expansion if further movement would backtrack:
\[
(\boldsymbol{\theta}^+ - \boldsymbol{\theta}^-)^\top \boldsymbol{\phi}^- < 0 \quad \text{or} \quad (\boldsymbol{\theta}^+ - \boldsymbol{\theta}^-)^\top \boldsymbol{\phi}^+ < 0,
\]
where $\boldsymbol{\theta}^-$ and $\boldsymbol{\theta}^+$ are the leftmost and rightmost positions on the current trajectory, and $\boldsymbol{\phi}^-$, $\boldsymbol{\phi}^+$ are their corresponding momenta.

Once a valid subtree is built, NUTS samples a point from it using a slice sampling-inspired procedure to preserve detailed balance.  For this paper, we use the \texttt{numpyro} module in Python to sample with NUTS. In the figure below, we see a demonstration of an MCMC following a Metropolis-Hastings algorithm and another following HMC with the NUTS sampler.
\begin{figure}[H]
    \centering
    \includegraphics[width=1\linewidth]{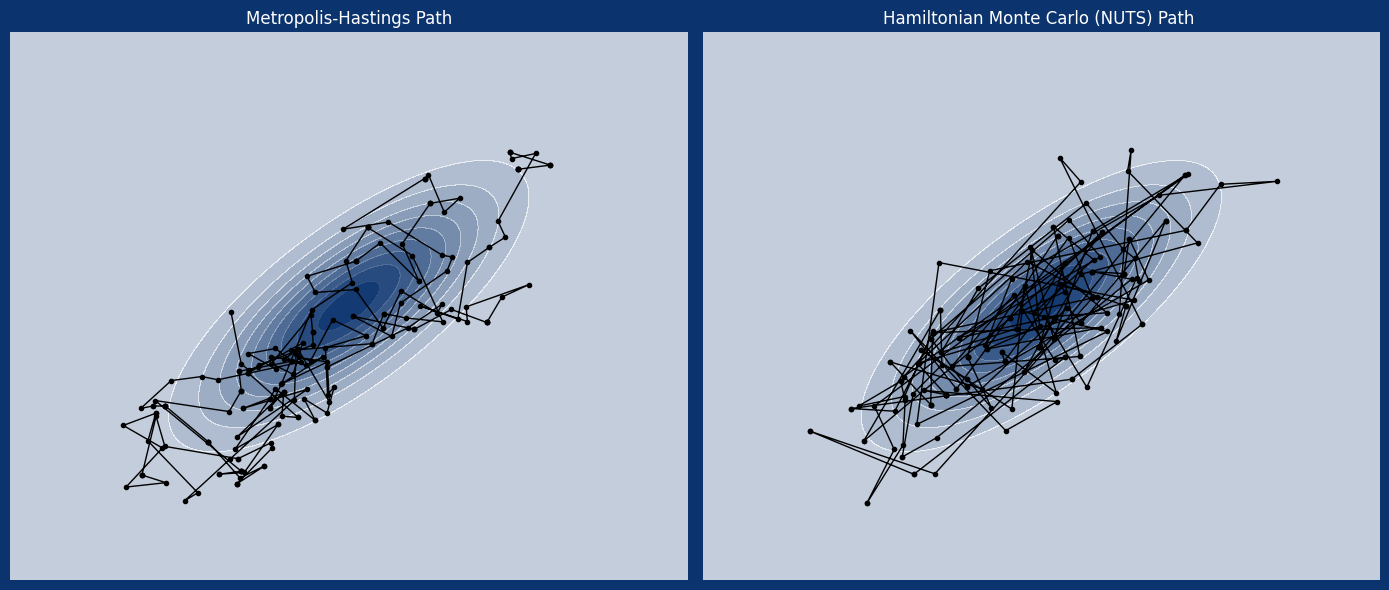}
    \caption{The path taken by a Metropolis-Hastings and Hamiltonian Monte Carlo MCMC on a bivariate Gaussian distribution with 200 samples drawn.}
    \label{fig:enter-label}
\end{figure}

\subsection{Prior Choice}
\label{sec:Prior}

In specifying prior distributions for the cosmological parameters in our models (\(\Lambda CDM\), \(wCDM\), \(CPL\)), we aimed to balance physical plausibility and constraints from previous observations. The chosen priors are weakly informative: they incorporate current knowledge about cosmological parameters while remaining sufficiently flexible to allow the data to speak for itself.

For the Hubble constant \(H_0\), we place a log-normal prior with mean 70 and standard deviation of scale 0.5 in log space. This choice ensures positivity, which is a physical necessity, and reflects the range of contemporary measurements such as those from the SH0ES collaboration \cite{Riess2022}, Planck 2018 \cite{Planck2018}, and the assumptions of the DES collaboration \cite{DES_SN5YR_data} when computing errors.

For the matter density parameter \(\Omega_m\), we adopt a Beta distribution with shape parameters 3 and 7, which places most of its mass in the physically plausible interval (0, 1), peaking around 0.3. This aligns well with current measurements, while being sufficiently uninformative to allow deviations. We prefer the Beta family here because it is naturally bounded on (0, 1), unlike the Normal distribution, and supports flexible shapes depending on the parameterization. 

In the \(wCDM\) model, the dark energy equation-of-state parameter \(w\) is assigned a normal distribution with mean \(-1\) and standard deviation \(2\). This centers the distribution on the \(\Lambda CDM\) value of  \(w= -1\), with ample room for variation in both directions. The width of \(2\) is consistent with typical prior ranges used in cosmological parameter estimation \cite{Sullivan2011} so that we do not artificially constrain the possibility of phantom energy or quintessence behavior. 

Similarly, in the CPL model, the time-varying equation of state is parameterized by \(w_0\) distributed by a normal distribution with mean \(-1\) and standard deviation \(2\), and \(w_a\) a normal distribution with mean \(0\) and standard deviation \(4\), consistent with results of previous studies using the CPL framework \cite{Zhao2017}, but wide enough to allow the data to have a significant influence.

 \subsection{Cholesky Whitening Transformation}
 \label{sec:Cholesky}

In the data provided by DES, the covariance matrix \(\Sigma\), exhibits a very high condition number, on the order of \(10^8.\) Inverting such an ill-conditioned matrix in likelihood evaluations can lead to severe numerical instability due to round-off error amplification and loss of precision. Whitening, in this case Cholesky whitening, which will be discussed shortly, replaces the unstable handling of \(\Sigma^{-1}\) with computationally stable triangular solves. Moreover, the log-determinant, \(\log |{\Sigma}|\) becomes numerically robust even for ill-conditioned matrices. From a geometric perspective, whitening “spheres” the likelihood surface into an isotropic domain, mitigating extreme anisotropy and improving conditioning for gradient-based sampling. While this does not automatically guarantee dramatic gains in effective sample size, it improves numerical stability of the posterior geometry and makes step-size adaptation in HMC more robust. Additionally, whitening allows us to decompose the joint likelihood into pointwise contributions, which leads to a more accurate Bayesian analysis and offers interpretable diagnostics via standardized residuals. 

\subsubsection{Whitening Process} 
In the data, the observational uncertainties are not independent, but correlated through the full covariance matrix, \(
\Sigma = \Sigma_{\mathrm{stat}} + \Sigma_{\mathrm{sys}} \in \mathbb{R}^{n \times n},
\) as defined before where \(n\) is the number of data points. To work with these correlations efficiently, we employ a whitening transformation based on the Cholesky factorization. Since $\Sigma$ is symmetric and positive definite, it admits a unique lower-triangular Cholesky factorization:
\[
\Sigma = L L^\top, \qquad L \in \mathbb{R}^{n \times n}.
\]
Given the residuals between data and model,
\[
r = \mu_{\mathrm{obs}} - \mu_{\mathrm{model}} \in \mathbb{R}^{n} ,
\]
we define the whitened residuals as
\[
y = L^{-1} r.
\]
This linear transformation has the effect of removing the correlations encoded in $\Sigma$. Intuitively, the transformation maps the original correlated error space into a new coordinate system in which the components are uncorrelated and have unit variance.  Conversely, to reconstruct correlated residuals from independent fluctuations $y \sim \mathcal{N}(0, I_n)$, we apply the inverse transformation:
\[
r = L y.
\]
\subsubsection{Justification}
Let $\mu_{\text{obs}}\in\mathbb{R}^{n}$ denote the observed data vector and $\theta\in\Theta\subset\mathbb{R}^d$ parameters with prior $f(\theta)$, where \(\Theta\) denotes the parameter space. Assume the likelihood
\[
\mu_{\text{obs}}\mid \theta \sim \mathcal{N}(\mu_{\text{model}}(\theta),\,\Sigma),
\]
where $\Sigma$ is a fixed, symmetric positive definite matrix independent of $\theta$ and of the model index $M$. Define the whitening map $T:=L^{-1}$ and the whitened residuals
\[
y(\theta):=T\big(\mu_{\text{obs}}-\mu_{\text{model}}(\theta)\big)\in\mathbb{R}^{n} .
\]
Denote by $\phi(\cdot)$ the univariate standard normal density and by $\Phi_n(\cdot;0,I)$ the standard $n$-variate normal density.

\begin{lemma}[Exact likelihood reparameterisation]
Let $L$ be the (lower-triangular) Cholesky factor of $\Sigma$, i.e.\ $\Sigma=LL^\top$, and define $T:=L^{-1}$. For
\[
r(\theta):=\mu_{obs}-\mu_{model}(\theta)
\quad\text{and}\quad
y(\theta):=T\,r(\theta)=L^{-1}\!\big(\mu_{obs}-\mu_{model}(\theta)\big),
\]
the likelihood satisfies
\[
f(\mu_{obs}\mid\theta)\;=\;|\det T|\;\Phi_n\!\big(y(\theta);\,0, I\big)
\;=\;|\det T|\;\prod_{i=1}^n \phi\!\big(y_i(\theta)\big),
\]
where $\Phi_n(\,\cdot\,;0,I)$ is the $n$-variate standard normal density and $\phi$ is the univariate standard normal density.
\end{lemma}

\begin{proof}
We begin the proof with the MVN density of the likelihood,
\begin{equation}\label{eq:mvn}
f(\mu_{obs}\mid\theta)
= \frac{1}{(2\pi)^{n/2}\,|\Sigma|^{1/2}}
\exp\!\left(-\tfrac12\, r(\theta)^\top \Sigma^{-1} r(\theta)\right),
\quad
r(\theta)=\mu_{obs}-\mu_{mdel}(\theta).
\end{equation}
Because $\Sigma$ is symmetric positive definite there exists a unique lower-triangular $L$ with positive diagonal such that $\Sigma=LL^\top$. Hence $\Sigma^{-1}=L^{-\top}L^{-1}$ and, writing $T:=L^{-1}$, we have
\begin{equation}\label{eq:quad}
r(\theta)^\top \Sigma^{-1} r(\theta)
= r(\theta)^\top L^{-\top}L^{-1} r(\theta)
= \big(L^{-1} r(\theta)\big)^\top \big(L^{-1} r(\theta)\big)
= y(\theta)^\top y(\theta),
\end{equation}
where $y(\theta):=T r(\theta)=L^{-1}r(\theta)$. Next, for the normalising constant note that
\begin{equation}\label{eq:det}
|\Sigma|=|LL^\top|=|L|^2
\quad\Longrightarrow\quad
|\Sigma|^{1/2}=|L|
\quad\Longrightarrow\quad
\frac{1}{|\Sigma|^{1/2}}=\frac{1}{|L|}=|\det L|^{-1}=|\det T|.
\end{equation}
Substituting \eqref{eq:quad} and \eqref{eq:det} into \eqref{eq:mvn} yields
\[
f(\mu_{obs}\mid\theta)
= \frac{1}{(2\pi)^{n/2}}\,|\det T|\,
\exp\!\left(-\tfrac12\, y(\theta)^\top y(\theta)\right).
\]
Recognising the $n$-variate standard normal density,
\[
\Phi_n\!\big(y(\theta);0,I\big)
= \frac{1}{(2\pi)^{n/2}}\exp\!\left(-\tfrac12\, y(\theta)^\top y(\theta)\right),
\]
we obtain
\[
f(\mu_{obs}\mid\theta)=|\det T|\;\Phi_n\!\big(y(\theta);0,I\big).
\]
Finally, since under $I$ the coordinates are independent, $\Phi_n(y;0,I)=\prod_{i=1}^n \phi(y_i)$, which gives
\[
f(\mu_{obs}\mid\theta)=|\det T|\;\prod_{i=1}^n \phi\!\big(y_i(\theta)\big).
\]
\end{proof}

\begin{lemma}[Posterior invariance under whitening]
The posterior density satisfies
\[
f(\theta\mid\mu_{\text{obs}}) \;\propto\; f(\theta)\,f(\mu_{\text{obs}}\mid\theta)
\;\propto\; f(\theta)\prod_{i=1}^n \phi\big(y_i(\theta)\big),
\]
hence the normalised posterior is identical whether expressed with the MVN likelihood or the whitened iid likelihood.
\end{lemma}
\begin{proof}
By the previous lemma $f(\mu_{\text{obs}}\mid\theta)=|\det T|\prod_i \phi(y_i(\theta))$, with $|\det T|$ independent of $\theta$. This multiplicative constant cancels in posterior normalisation.
\end{proof}

\begin{theorem}[MH invariance under whitening]
Let $\tilde\pi(\theta)$ be the unnormalised posterior density based on the original MVN likelihood and prior $f(\theta)$, and let $\tilde\pi_w(\theta)$ be the unnormalised posterior based on the whitened independent, identically distributed likelihood and the same prior. Then there exists a constant $c>0$ independent of $\theta$ such that $\tilde\pi(\theta)=c\,\tilde\pi_w(\theta)$ for all $\theta$.
Consequently, for any proposal kernel $q(\theta'\mid\theta)$, the Metropolis--Hastings acceptance probability is identical when targeting $\tilde\pi$ or $\tilde\pi_w$, and the induced MH Markov chain transition kernels coincide.
\end{theorem}
\begin{proof}
By the MVN density and change of variables,
\[
f(\mu_{\text{obs}}\mid\theta)
=\frac{1}{(2\pi)^{n/2}|\Sigma|^{1/2}}
\exp\!\Big(-\tfrac12(\mu_{\text{obs}}-\mu_{\text{model}}(\theta))^\top \Sigma^{-1}(\mu_{\text{obs}}-\mu_{\text{model}}(\theta))\Big)
=|\det T|\prod_{i=1}^n \phi\big(y_i(\theta)\big),
\]
with $|\det T|=|\det L|^{-1}$ independent of $\theta$.
Thus
\[
\tilde\pi(\theta)=f(\theta)\,f(\mu_{\text{obs}}\mid\theta)
= \big(|\det T|\big) \; f(\theta)\prod_{i=1}^n \phi\big(y_i(\theta)\big)
=: c\,\tilde\pi_w(\theta),
\]
with $c=|\det T|>0$ constant in $\theta$.
For MH with proposal $q$, the acceptance probability is
\[
\alpha(\theta,\theta')
= \min\Big\{1,\;\frac{\tilde\pi(\theta')\,q(\theta\mid\theta')}{\tilde\pi(\theta)\,q(\theta'\mid\theta)}\Big\}.
\]
Replacing $\tilde\pi$ by $c\,\tilde\pi$ multiplies the numerator and denominator by the same constant $c$, leaving the ratio unchanged. Therefore $\alpha$ is identical, and so are the transition kernels.
\end{proof}
\begin{theorem}[HMC flow invariance under whitening]
Let the posterior for parameters $\theta\in\mathbb{R}^d$ be proportional to
\[
f(\theta)\,f(\mu_{\mathrm{obs}}\mid\theta),
\]
where the data likelihood is multivariate normal $\mu_{\mathrm{obs}}\mid\theta\sim\mathcal{N}(\mu_{\mathrm{model}}(\theta),\Sigma)$ with $\Sigma\in\mathbb{R}^{n\times n}$ symmetric positive definite and independent of $\theta$. Let $L$ be the Cholesky factor of $\Sigma$ and set $T:=L^{-1}$. Define the potentials
\[
U(\theta):=-\log f(\theta)-\log f(\mu_{\mathrm{obs}}\mid\theta) + C,
\qquad
U_w(\theta):=-\log f(\theta)-\sum_{i=1}^n\log\phi\big(y_i(\theta)\big) + D,
\]
where $y(\theta)=T(\mu_{\mathrm{obs}}-\mu_{\mathrm{model}}(\theta))$ and $\phi$ is the univariate standard normal density .
Then $\log|\det T|$ is constant in $\theta$ and
\[
U(\theta)=U_w(\theta)-\log|\det T|.
\]
Consequently, $\nabla_\theta U(\theta)=\nabla_\theta U_w(\theta)$, Hamiltonian trajectories (in exact arithmetic) coincide, Metropolis acceptance probabilities are identical, and the ideal HMC/NUTS transition kernels coincide.
\end{theorem}
\begin{proof}
By the multivariate normal density and the whitening change of variables,
\[
\log f(\mu_{\mathrm{obs}}\mid\theta)=\log|\det T|+\sum_{i=1}^n\log\phi\big(y_i(\theta)\big).
\]
Substituting into $U(\theta)$ gives
\[
U(\theta)=-\log f(\theta)-\log|\det T|-\sum_{i=1}^n\log\phi\big(y_i(\theta)\big)
=U_w(\theta)-\log|\det T|.
\]
Since $|\det T|$ does not depend on $\theta$ by assumption, it follows that
$\nabla_\theta U(\theta)=\nabla_\theta U_w(\theta)$. Hamiltonian dynamics (leapfrog updates) depend only on these gradients and on the kinetic energy $K(p)$, which is independent of the likelihood transformation; hence the continuous-time trajectories coincide. The Metropolis accept/reject step uses differences of the Hamiltonian $H=U+K$, in which the constant $-\log|\det T|$ cancels, giving identical acceptance probabilities. Therefore the exact (ideal) HMC/NUTS transition kernels coincide.
\end{proof}

The theorem above can be generalised to a diffeomorphism between two manifolds.  From the same, we know that the leapfrog integrator is symplectic and reversible for separable Hamiltonians. For more details see \cite{José_Saletan_1998, neal2011mcmc, betancourt2017conceptual}.

\begin{theorem}[Equivariance of Leapfrog under Affine Transformations]
Let $T:\mathbb{R}^d \to \mathbb{R}^d$ be an affine map of the form
\[
T(\xi) = J\xi + b,
\]
where $J$ is an invertible constant matrix and $b$ a translation vector.  
Suppose Hamiltonian dynamics are simulated with leapfrog using potential energy $U(\theta)$ and kinetic energy
\[
K_\theta(p_\theta) = \tfrac12\, p_\theta^\top M^{-1} p_\theta.
\]
Define new coordinates $\xi = J^{-1}(\theta-b)$ with corresponding momenta $p_\xi = J^\top p_\theta$.  
Then the leapfrog integrator is \emph{equivariant} under $T$: a leapfrog step of size $\varepsilon$ in $(\xi,p_\xi)$-coordinates corresponds exactly to a leapfrog step of the same size in $(\theta,p_\theta)$-coordinates.
\end{theorem}

\begin{proof}
Since $J$ is constant, the momentum transformation $p_\theta = J^{-T} p_\xi$ is independent of position, and the induced kinetic energy in $\xi$-coordinates,
\[
K_\xi(p_\xi) = \tfrac12\, p_\xi^\top (J^{-1} M^{-1} J^{-T}) p_\xi,
\]
is quadratic with constant metric. The leapfrog scheme consists of position updates using $\nabla K$ and momentum updates using $\nabla U$. Under a linear reparametrisation, these updates transform linearly with no additional dependence on $\xi$. Therefore applying leapfrog in $\xi$ and then mapping by $T$ yields exactly the same update as mapping by $T$ first and applying leapfrog in $\theta$.  
\end{proof}

In fact, the above theorem can be generalised to show that this is an if and only if case, however this is out of scope of this paper.

In the figures below, we display the results of running four independent HMC chains with 9000 samples after warm-up for all three cosmological models.

\begin{figure}[H]
    \centering
    \includegraphics[width=1\linewidth]{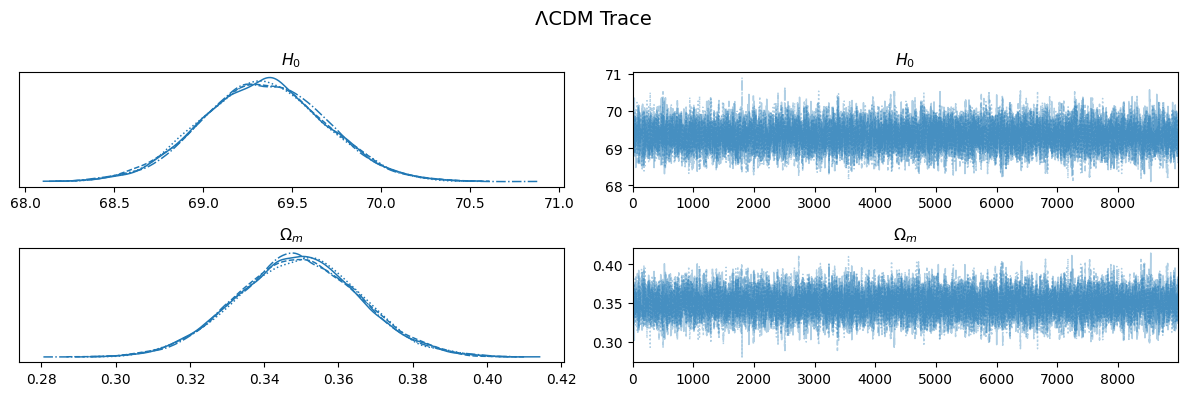}
    \caption{Trace plots for posterior samples of \(H_0\) and \(\Omega_m\)under the \( \Lambda CDM\) model. The two set of curves on the left show the posteriors estimated by four independent HMC chains. Each chain shows stable mixing and no visible divergences in the trace plots on the right, indicating good convergence of the NUTS sampler in NumPyro.}
    \label{fig:lcdm_trace}
\end{figure}
\begin{figure}[H]
    \centering
    \includegraphics[width=1\linewidth]{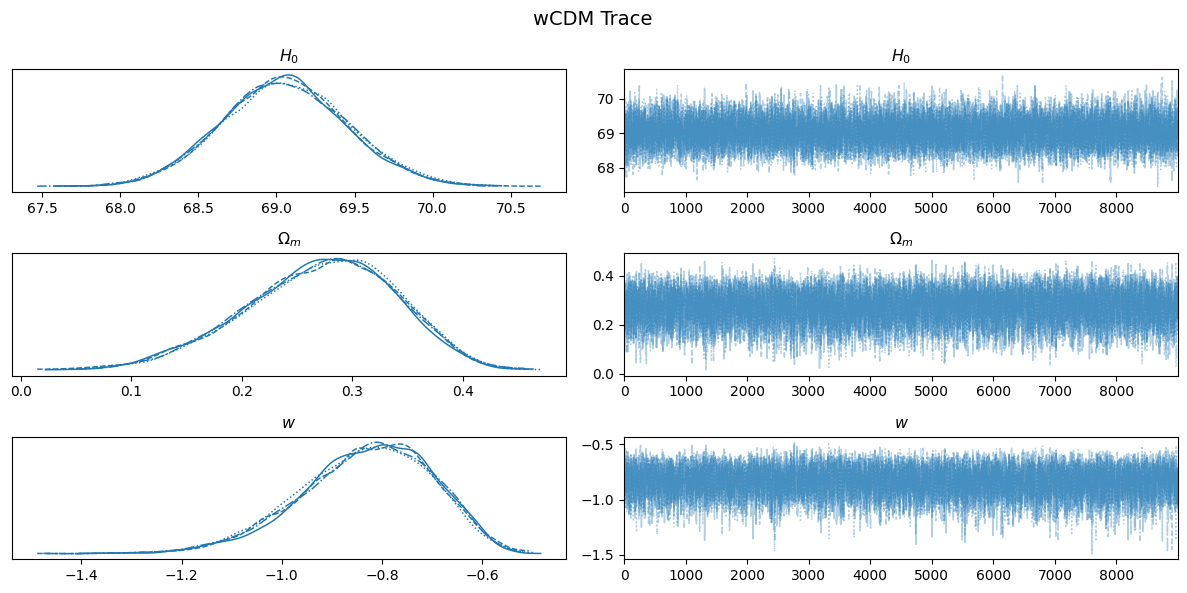}
    \caption{Same as Fig. 3, but for the \(wCDM\) model. The chains again show good convergence, with slight skewing for the \(w\) parameter's trace. }
    \label{fig:wcdm_trace}
\end{figure}
\begin{figure}[H]
    \centering
    \includegraphics[width=1\linewidth]{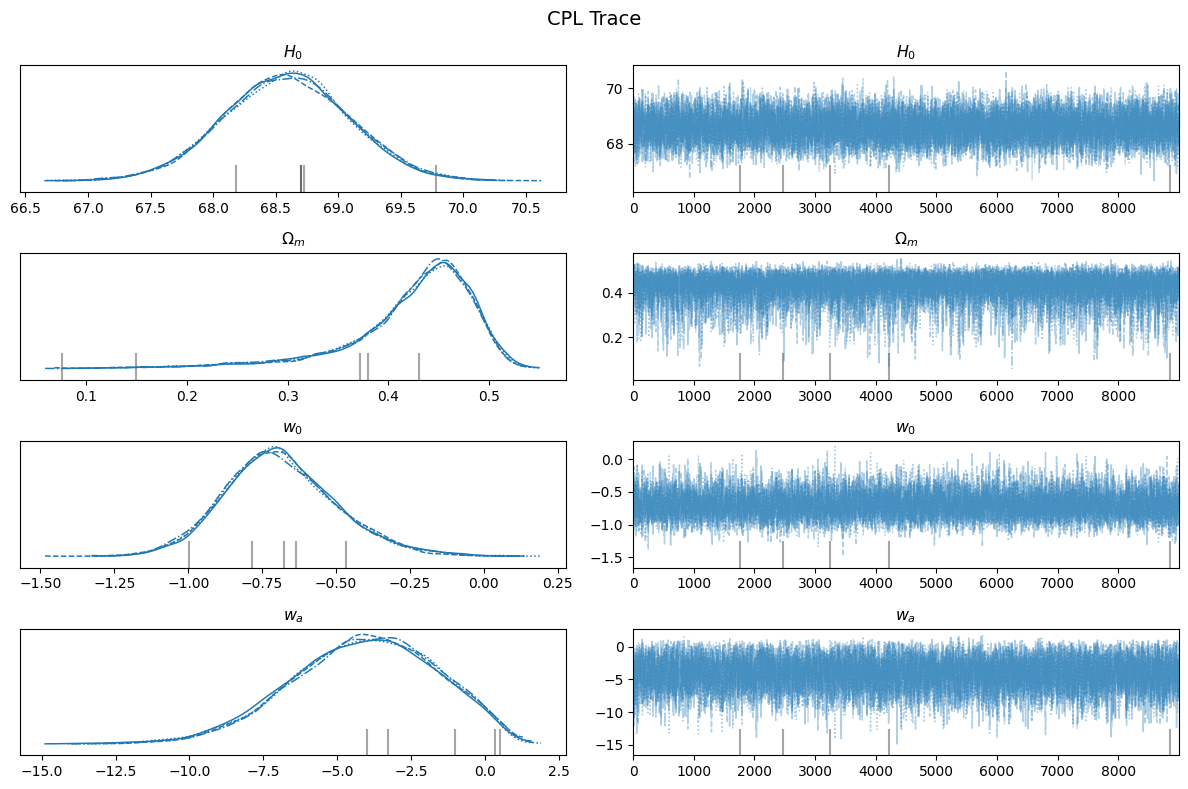}
    \caption{Same as Fig. 3, but for the \(CPL\) model. The chains again show good convergence, with skewness for the \(\Omega_m\) parameter's trace and similar skewing for the \(w_a\) parameter as in Fig. 4 for \(w\).}
    \label{fig:CPL_trace}
\end{figure}

\subsection{Model Selection and Evaluation Methods}
\label{sec:Model_Select}

\subsubsection{Bayes Factors}
Bayes factors provide a framework to quantify relative evidence that data provide for competing hypotheses within the Bayesian paradigm. For models $M_1$ and $M_2$ and data \(D\), the Bayes factor is defined as the ratio of marginal likelihoods:
\begin{equation}
BF_{12} = \frac{f(D | M_1)}{f(D | M_2)}
\end{equation}
where each marginal likelihood represents the evidence for model $k$, as was discussed previously. Through Bayes' theorem, Bayes factors directly relate to posterior model probabilities. If models possess equal prior probabilities $\mathbb{P}(M_1) = \mathbb{P}(M_2)$, then:
\begin{equation}
\frac{\mathbb{P}(M_1 | D)}{\mathbb{P}(M_2 | D)} = BF_{12}
\end{equation}
More generally, the posterior odds equal the product of prior odds and the Bayes factor:
\begin{equation}
\frac{\mathbb{P}(M_1 | y)}{\mathbb{P}(M_2 | y)} = BF_{12} \cdot \frac{\mathbb{P}(M_1)}{\mathbb{P}(M_2)}
\end{equation}
Computing marginal likelihoods presents significant computational challenges; as such, we use the bridge sampling approximation to mend this. The method introduces auxiliary densities that bridge between prior and posterior, providing more stable estimates through the relationship:
\begin{equation}
\frac{f(D | M_1)}{f(D | M_2)} = \frac{\mathbb{E}_{\text{post}_2}[f(D | \theta_1, M_1) f(\theta_1 | M_1) / q(\theta_1)]}{\mathbb{E}_{\text{post}_1}[f(D | \theta_2, M_2) f(\theta_2 | M_2) / q(\theta_2)]}
\end{equation}
where $q(\theta)$ represents a bridging density, and the expectations are taken under the posterior distributions of the respective models with \(\theta_i\)-s representing the respective inferred parameters. For this paper, we used the \texttt{bridgesampling} library in R to conduct bridge sampling to estimate the Bayes factors \cite{Gronau2017}.

\subsubsection{Widely Applicable Information Criterion}
The Widely Applicable Information Criterion (WAIC) \cite{Watanabe2010}, represents a fully Bayesian extension of the Akaike Information Criterion (AIC). Unlike AIC, which conditions on point estimates, WAIC properly accounts for parameter uncertainty by averaging over the entire posterior distribution. Unlike classical information criteria, WAIC accounts for model complexity and is particularly well-suited for complex models such as those in this paper. It is known that theoretically, the leave-one-out cross-validation in Bayes estimation is asymptotically equal to the widely applicable information criterion and that the sum of the cross-validation error \cite{Watanabe2010};  therefore, we use WAIC as a sufficient tool to measure how well the three models predict observed values of \(\mu\). We compute WAIC using the \texttt{arviz} module in Python, which estimates the expected log pointwise predictive density (elpd) and its variance directly from posterior samples.

\begin{lemma}[WAIC factorisation after whitening]
Let  $\ell_i(\theta):=\log\phi(y_i(\theta))$ be the pointwise log-likelihood contributions after whitening by \(T\). Then
\[
\log f(\mu_{\text{obs}}\mid\theta)=\log|\det T|+\sum_{i=1}^n \ell_i(\theta),
\]
where $\log|\det T|$ is constant in $\theta$.
\end{lemma}
\begin{proof}
The identity for the log-likelihood follows immediately from the reparameterised form:
\[
\log f(\mu_{\text{obs}}\mid\theta) = \log|\det T| + \sum_{i=1}^n \ell_i(\theta).
\]

\end{proof}
\begin{remark}
The whitened likelihood admits an exact factorisation into conditionally independent terms. Consequently, WAIC computed from $\{\ell_i(\theta)\}_{i=1}^n$ are valid for the factorisation induced by $T$. WAIC uses the pointwise contributions $\ell_i(\theta)$ via
\[
\mathrm{lppd} = \sum_{i=1}^n \log\Big(\frac{1}{S}\sum_{s=1}^S e^{\ell_i(\theta^{(s)})}\Big), \qquad
p_{\text{waic}} = \sum_{i=1}^n \mathrm{Var}_s\big[\ell_i(\theta^{(s)})\big],
\]
and
\[
\mathrm{elpd}_{\text{waic}} = \mathrm{lppd} - p_{\text{waic}}.
\]

The additive constant $\log|\det T|$ does not depend on $\theta$; it cancels in the variance term and adds the same constant to lppd across models sharing $T$, so differences in $\mathrm{elpd}$ are preserved. For more information on WAIC computation, see ~\cite{gelman2014}.
\end{remark}

\subsubsection{Prior and Posterior Predictive Checks}
To ensure that our models are reasonable and well-specified, we performed both prior and posterior predictive checks. Prior predictive checks involve generating hypothetical data from the model using only the prior distributions on the parameters, before observing the actual dataset. This allows us to verify that the priors and the model structure are sensible and capable of producing plausible data. Posterior predictive checks are performed after fitting the model to the observed data: we generate replicated datasets using parameter values drawn from the posterior distribution. By comparing these replicated datasets to the observed data, we can assess whether the model is capable of capturing the key features of the data. Together, these checks provide a practical way to diagnose potential misspecifications and ensure that the models are suitable for subsequent comparisons using WAIC or Bayes factors.

\begin{lemma}[Prior and posterior predictive equivalence between scales]
Let $\theta \sim f(\theta)$ and $y \sim \mathcal{N}(0, I_n)$ be independent. Define
\[
\mu := \mu_{\rm model}(\theta) + L y.
\]
Then:
\begin{enumerate}
\item (\emph{Prior predictive}) Conditionally on $\theta$, $\mu \sim \mathcal{N}(\mu_{\rm model}(\theta), \Sigma)$. Marginally over $\theta$, this recovers the usual prior predictive distribution.
\item (\emph{Posterior predictive}) If $\theta \sim f(\theta \mid \mu_{\rm obs})$ and $y' \sim \mathcal{N}(0, I_n)$ independent, then
\[
\mu' := \mu_{\rm model}(\theta) + L y'
\]
has the usual posterior predictive distribution. Consequently, any test statistic computed on simulated $\mu$ is identical whether generated in whitened space and mapped back via $L$ or generated directly in the original scale.
\end{enumerate}
\end{lemma}
\begin{proof}
Since $y \sim \mathcal{N}(0, I_n)$, a linear transformation gives
\[
Ly \sim \mathcal{N}(0, LL^\top) = \mathcal{N}(0, \Sigma),
\]
so conditionally on $\theta$,
\[
\mu = \mu_{\rm model}(\theta) + Ly \sim \mathcal{N}(\mu_{\rm model}(\theta), \Sigma),
\]
proving the prior predictive result.

For the posterior predictive, the same argument applies: conditionally on $\theta$,  $\mu' = \mu_{\rm model}(\theta) + Ly' \sim \mathcal{N}(\mu_{\rm model}(\theta), \Sigma)$, and marginalising over the posterior $f(\theta \mid \mu_{\rm obs})$ recovers the usual posterior predictive distribution. By the posterior invariance lemma, simulating in the whitened space and mapping back yields identical distributions for any test statistic.
\end{proof}

\begin{figure}[H]
    \centering

        \includegraphics[width=\linewidth]{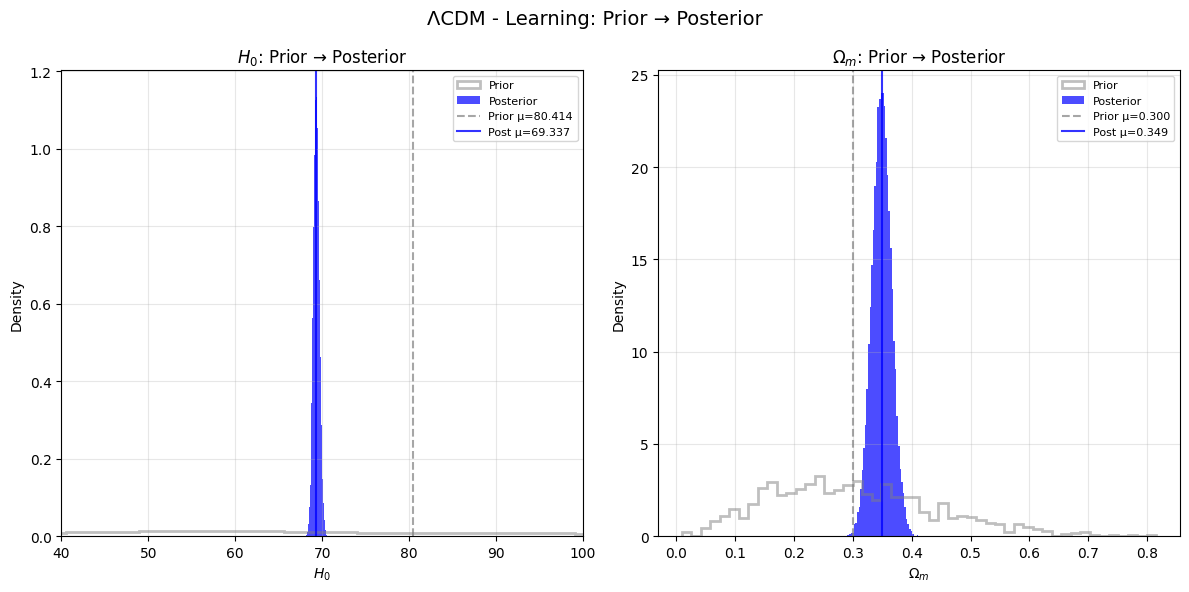}
        
    \caption{Comparison of prior (gray hollow curves) and posterior (blue solid curves) distributions for the parameters \(H_0\) and \(\Omega_m\) under the \(\Lambda\)CDM model. Priors are weakly informative over physically reasonable ranges. The posterior distributions show significant shrinkage and change in shape relative to the priors, indicating that the DES-SN5YR data provided strong constraints on both parameters.}
    \label{fig:lcdm_side_by_side}
\end{figure}

\begin{figure}[H]
    \centering

        \includegraphics[width=\linewidth]{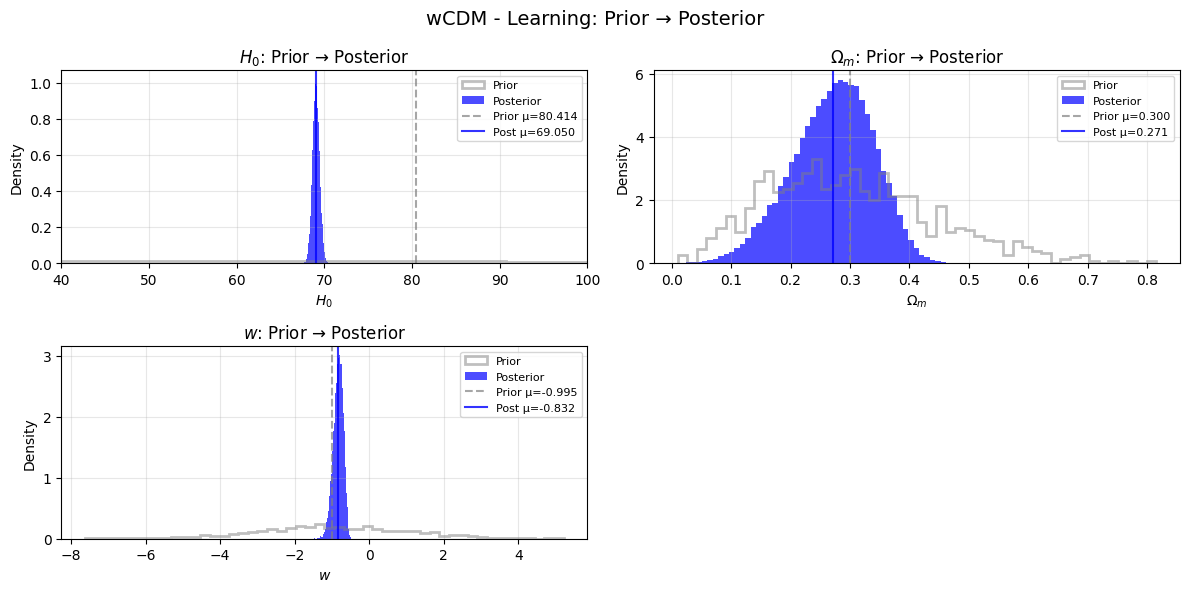}
        
    \caption{Same as Fig. 6, but for the \(w\)CDM model. The posterior for \(H_0\) remains largely consistent with the \(\Lambda\)CDM case, while \(\Omega_m\) is more weakly constrained by the data. The additional free parameter \(w\) is moderately well constrained, with its posterior indicating a preference slightly above the cosmological constant value \(w = -1\). The broader posterior for \(\Omega_m\) reflects the mild degeneracy between matter density and the dark-energy equation-of-state parameter.
}
    \label{fig:wcdm_side_by_side}
\end{figure}

\begin{figure}[H]
    \centering

        \includegraphics[width=\linewidth]{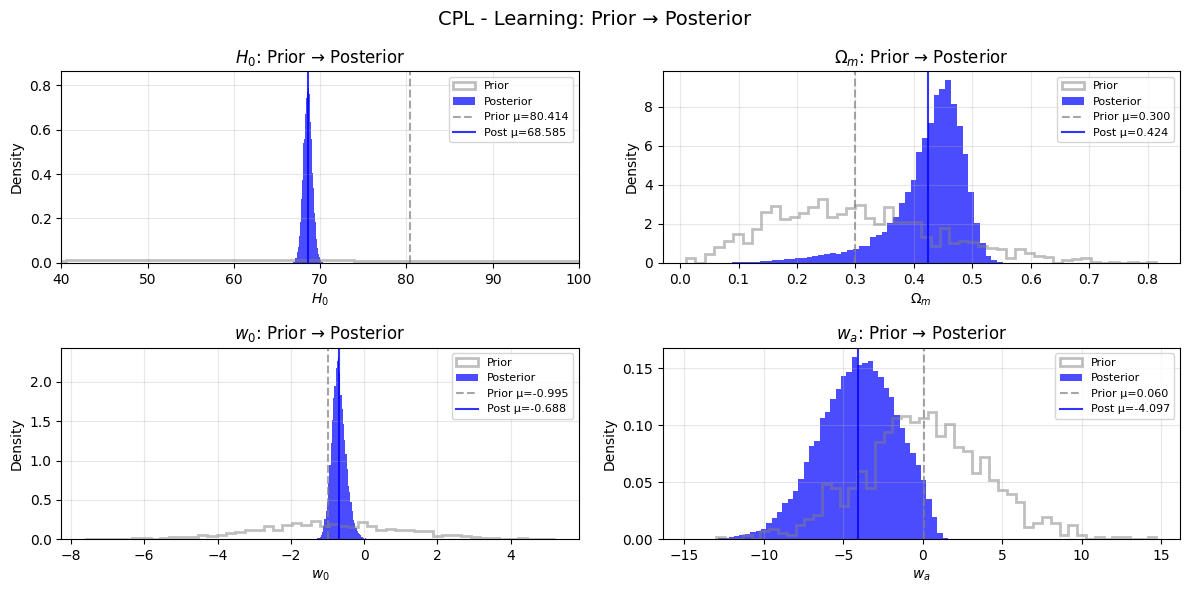}
        
    \caption{Same as Fig. 7, but for the \(CPL\) model. The posterior for \(H_0\) remains largely consistent with the \(\Lambda\)CDM case, while \(\Omega_m\) is more weakly constrained by the data. The additional parameters \(w_0\) and \(w_a\) show differing levels of constraint. \(w_0\) is relatively well constrained, with its posterior indicating a mild deviation from the cosmological constant value \(w_0 = -1\), while \(w_a\) remains weakly constrained. The broader posteriors for \(\Omega_m\) and \(w_a\) reflect the degeneracy between matter density and the evolving dark-energy equation-of-state parameters.}
    \label{fig:cpl_side_by_side}
\end{figure}

To quantify the degree to which our inferences are driven by the data rather than by prior specification, we computed the marginal Kullback--Leibler (KL) divergence \cite{KullbackLeibler1951} between the posterior and prior distributions for each parameter. For a parameter $\theta$, this is defined as
\begin{equation}
D_{\mathrm{KL}}(\pi(\theta) \,|\, p(\theta)) \;=\; \int \pi(\theta) \, \log \frac{\pi(\theta)}{p(\theta)} \, d\theta,
\end{equation}
where $\pi(\theta)$ is the posterior density and $p(\theta)$ the prior density \cite{Shlens2014, JewsonSmithHolmes2018}. 

We approximate the posterior density $\pi(\theta)$ using a Gaussian kernel density estimate (KDE) using the NumPyro library in Python. The prior density $p(\theta)$ is available analytically from the NumPyro distribution objects used in the models (LogNormal for $H_0$, Beta for $\Omega_m$). The densities are normalized and evaluated on a fine grid of 2000 points spanning the support of the posterior (approximately $\pm 6\sigma$ around the mean), and the KL divergence is then computed numerically by trapezoidal quadrature. 

Alongside KL divergence, we also report a simple shrinkage measure of posterior relative to prior standard deviation:
\begin{equation}
\text{shrinkage} \;=\; 1 - \frac{\sigma_{\text{post}}}{\sigma_{\text{prior}}},
\end{equation}
which quantifies the extent to which the posterior distribution is narrower than the prior. KL divergence captures how much the shape of the posterior departs from the prior, while shrinkage quantifies how much the scale of uncertainty has reduced. We measure KL values using the natural logarithm (nats). Larger nats  and high shrinkage indicate that the data are strongly informative relative to the prior, while very low values would suggest that posterior inference is substantially prior-driven. We note that these KL values are marginal and do not capture dependencies between parameters. We also note that the measurements are numerical approximations. Below is a table with the results from the inference in this paper.

\begin{table}[H]
\centering
\footnotesize
\begin{tabular}{llrrrrrr}
\toprule
Model & Parameter & Prior $\mu$ & Prior $\sigma$ & Posterior $\mu$ & Posterior $\sigma$ & Shrinkage & KL (nats) \\
\midrule
$\Lambda$CDM& $H_{0}$     & 79.32 & 42.27 & 69.34 & 0.35 & 99.2\% & 1.42 \\
& $\Omega_{m}$ & 0.300 & 0.138 & 0.349 & 0.017 & 88.0\% & 1.35 \\
\midrule
$w$CDM& $H_{0}$     & 79.32 & 42.27 & 69.05 & 0.40 & 99.0\% & 1.41 \\
& $\Omega_{m}$ & 0.300 & 0.138 & 0.271 & 0.068 & 50.9\% & 0.36\\
& $w$         & -1.00 & 2.00  & -0.83 & 0.13 & 93.4\% & 1.38 \\
\midrule
CPL& $H_{0}$     & 79.32 & 42.27 & 68.59 & 0.50 & 98.8\%& 1.42 \\
& $\Omega_{m}$ & 0.300 & 0.138 & 0.424 & 0.064 & 53.8\% & 1.06\\
& $w_{0}$     & -1.00 & 2.00  & -0.69 & 0.19 & 90.5\% & 1.33 \\
& $w_{a}$     & 0.00  & 4.00  & -4.10 & 2.44 & 39.0\% & 0.72 \\
\bottomrule
\end{tabular}
\caption{Prior--posterior comparison diagnostics for $\Lambda$CDM, $w$CDM, and CPL models. 
Reported are prior and posterior means ($\mu$), standard deviations ($\sigma$), percentage shrinkage in variance, and approximate KL divergences (posterior $|$ prior) in nats.}
\label{tab:prior-posterior-diagnostics}
\end{table}

\section{Results and Analysis}

In this section we present the main findings of our Bayesian inference of cosmological parameters from the DES-SN5YR Type Ia supernova dataset which was expanded on in  \hyperref[sec:data]{Section 3}. The results of performing the Hamiltonian Monte Carlo (HMC) as described in \hyperref[sec:HMC]{Section 4.2} with the relevant whitening described in \hyperref[sec:Cholesky]{Section 4.5} on the cosmological models discussed in \hyperref[sec:cosmo_models]{Section 2.2}: \(\Lambda\)CDM, wCDM, and CPL. Each model's parameters were sampled with four independent chains each consisting of 10,000 steps. 

We begin the analysis by examining the convergence of the HMC chains themselves.  Examining the trace plots in \hyperref[fig:lcdm_trace]{Figure 3}, \hyperref[fig:wcdm_trace]{Figure 4} , and \hyperref[fig:CPL_trace]{Figure 5} indicates good mixing with no visible trends, suggesting that the chains have converged successfully.

We now look at the $\hat{r}$ values in \hyperref[tab:lcdm-des-summary]{Table 2}, \hyperref[tab:wcdm-des-summary]{Table 3}, \hyperref[tab:cpl-des-summary]{Table 4}. The $\hat{r}$ values are essentially unity for every parameter, with maximum deviations of less than $0.001$. This indicates excellent mixing of the chains and the absence of non-converged modes. Similarly, the effective sample sizes (ESS) are comfortably above $6,000$ in all cases, with most parameters exceeding $8,000$–$12,000$. These numbers far surpass the common rule-of-thumb threshold of $1,000$, suggesting that Monte Carlo error is negligible relative to posterior uncertainty.

The Monte Carlo standard errors (MCSEs) for both the mean and standard deviation are correspondingly small: at most a few $10^{-3}$ for the location parameters, and well below $0.05$ even for the weakly constrained CPL parameter $w_a$. This shows that posterior summaries such as means and HDIs are based on stable estimates, not noisy chain averages. Taken together, these diagnostics demonstrate that the chains have converged robustly and that posterior quantities reported in the parameter tables represent well-sampled distributions. 

Beyond standard convergence diagnostics, we also examine the influence of the prior assumptions on the inferred posteriors by comparing prior and posterior distributions directly. \hyperref[tab:prior-posterior-diagnostics]{Table 1} reports prior and posterior means, standard deviations, percentage shrinkage, and approximate Kullback--Leibler (KL) divergences $D_{\mathrm{KL}}(\pi(\theta)\,|\,p(\theta))$, all of which was described in \hyperref[sec:Model_Select]{Section 4.6.3}. In all three cosmological models, the posterior uncertainties are  reduced relative to the priors, with shrinkage values exceeding $90\%$ for most parameters. This indicates that the constraints are driven almost entirely by the data rather than by prior informativeness. The KL divergences, which quantify the information gain from prior to posterior, are mostly of order unity in natural units ( \(\sim 1 - 1.5 \text{ nats} \)), further supporting the conclusion that the data provide substantial updating of the prior beliefs. Parameters that are less tightly constrained, such as $w_a$ in CPL and $\Omega_m$ in $w$CDM, show reduced shrinkage and smaller KL divergences, reflecting the limited constraining power of the supernova dataset in those directions. 

\begin{table}[H]
    \centering
    \resizebox{\textwidth}{!}{
    \begin{tabular}{llllllllcc}\toprule
   & mean& standard deviation&HDI \(3\%\)& HDI \(97\%\)& MCSE mean&MCSE standard deviation&  ESS bulk&ESS tail&\(\hat{r}\)\\\midrule
            \(H_0\)& 69.33687  &  0.34543  & 68.70442    &  69.99445      &   0.00394     & 0.00291    & 7676.06076 &8804.30733   &  1.00045   \\
 \(\Omega_m\)&  0.34919   & 0.01659    & 0.31738     & 0.37999       & 0.00019  & 0.00014  & 7554.19472    & 7827.01680   & 1.00044    \\ \bottomrule\end{tabular}
 }
    \caption{Summary for the \(\Lambda\)CDM inference}
    \label{tab:lcdm-des-summary}
\end{table}
\begin{table}[H]
    \centering
    \resizebox{\textwidth}{!}{
    \begin{tabular}{llllllllcc}
    \toprule
    & mean & standard deviation & HDI \(3\%\) & HDI \(97\%\) & MCSE mean & MCSE standard deviation & ESS bulk & ESS tail & \(\hat{r}\) \\
    \midrule
    \(H_0\)      & 69.05035    & 0.40305     & 68.29492    & 69.81798     & 0.00387    &0.00268   &   10872.12957   & 12434.70591   & 1.00031  \\
    \(\Omega_m\) & 0.27099     & 0.06788      & 0.14265      & 0.39405      & 0.00075    & 0.00051    & 8410.76372   & 9001.69727    & 1.00084   \\
    \(w\)        & -0.83235    & 0.13178     & -1.07945     & -0.59912     & 0.00146    & 0.00102    & 7934.64946   & 9404.79091   & 1.00090  \\
    \bottomrule
    \end{tabular}
    }
    \caption{Summary for the wCDM inference}
    \label{tab:wcdm-des-summary}
\end{table}
\begin{table}[H]
    \centering
    \resizebox{\textwidth}{!}{
    \begin{tabular}{llllllllcc}
    \toprule
    & mean & standard deviation & HDI \(3\%\) & HDI \(97\%\) & MCSE mean & MCSE standard deviation & ESS bulk & ESS tail & \(\hat{r}\) \\
    \midrule
    \(H_0\)      & 68.58548    & 0.50046     & 67.64297     & 69.51714     & 0.00539    & 0.00338    & 8611.34281   & 11858.42566   & 1.00010  \\
    \(\Omega_m\) & 0.42443     & 0.06383     & 0.30055      & 0.52239      & 0.00084    & 0.00104    & 7175.11622   & 6110.96641    & 1.00012  \\
    \(w_0\)      & -0.68847    & 0.19080     & -1.03612     & -0.31275     & 0.00224    & 0.00163    & 7474.27764   & 8798.41156   & 1.00038   \\
    \(w_a\)      & -4.09688    & 2.43860     & -8.37493     & 0.43779      & 0.03112    & 0.01900    & 6031.27212   & 8280.06544   & 1.00038 \\
    \bottomrule
    \end{tabular}
    }
    \caption{Summary for the CPL inference}
    \label{tab:cpl-des-summary}
\end{table}

Our inference strategy differs from that used in the DES-SN5YR analysis \cite{DES2024}. In our models, the Hubble constant  \(H_0\) was sampled explicitly with a prior. In contrast, the DES collaboration marginalised over the absolute magnitude–Hubble constant degeneracy, such that their supernova-only analysis does not yield an independent constraint on \(H_0\) \cite{DES2024}. This distinction means that our reported \(H_0\) posteriors should be interpreted as conditional on the chosen prior, rather than as stand-alone measurements. Despite this difference in nuisance treatment, the parameter estimates reported in \hyperref[tab:lcdm-des-summary]{Table 1}, \hyperref[tab:wcdm-des-summary]{Table 2}, and \hyperref[tab:lcdm-des-summary]{Table 3} are broadly consistent with those reported by DES and earlier supernova studies \cite{Brout2022, Planck2018}

\begin{figure}[H]
    \centering
    \includegraphics[width=0.65\linewidth]{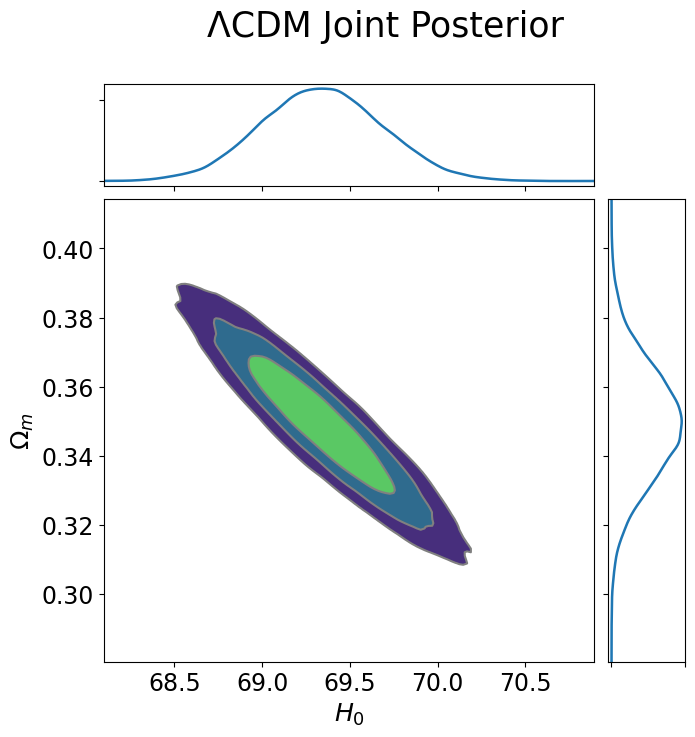}
    \caption{Posterior distributions for cosmological parameters under the \(\Lambda\)CDM model. The contours represent the three default regions obtained from Arviz after using the NUTS sampler in NumPyro, based on 9000 posterior samples from the DES-SN5YR dataset. The parameters include the Hubble constant \(H_0\) and matter density \(\Omega_m\). We see the degeneracy between \(H_0\) and \(\Omega_m\), as increases in \(H_0\) can be partially compensated by decreases in \(\Omega_m\), resulting in nearly identical luminosity distance–redshift relations.}
    \label{fig:lcdm_posterior}
\end{figure}

\begin{figure}[H]
    \centering
    \includegraphics[width=0.65\linewidth]{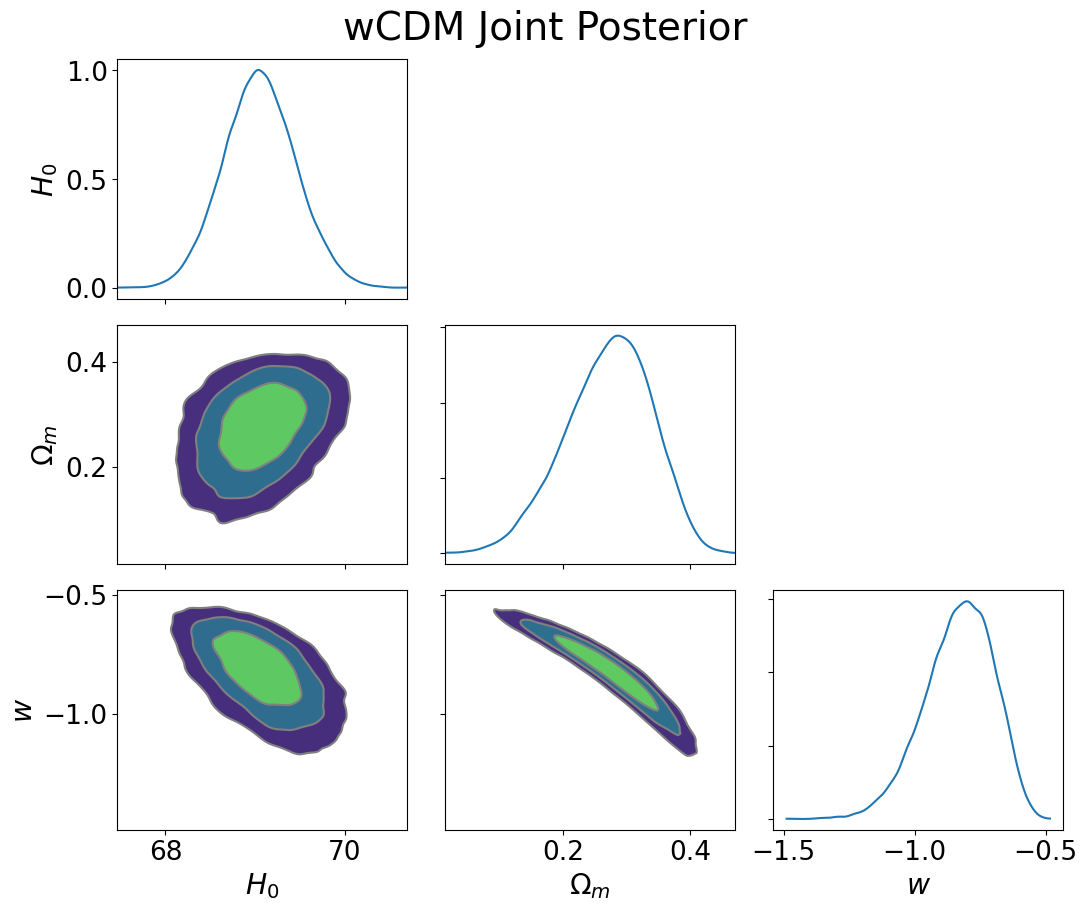}
    \caption{Posterior distributions for cosmological parameters under the \(w\)CDM model are structured similarly to Fig. 9. The parameters include the Hubble constant \(H_0\), matter density \(\Omega_m\), and dark-energy equation-of-state parameter \(w\). Compared to the \(\Lambda\)CDM case, the posterior exhibits a degeneracy between \(w\) and \(\Omega_m\), reflecting the additional flexibility introduced by the free dark-energy parameter.}
    \label{fig:wcdm_posterior}
\end{figure}

\begin{figure}[H]
    \centering
    \includegraphics[width=0.65\linewidth]{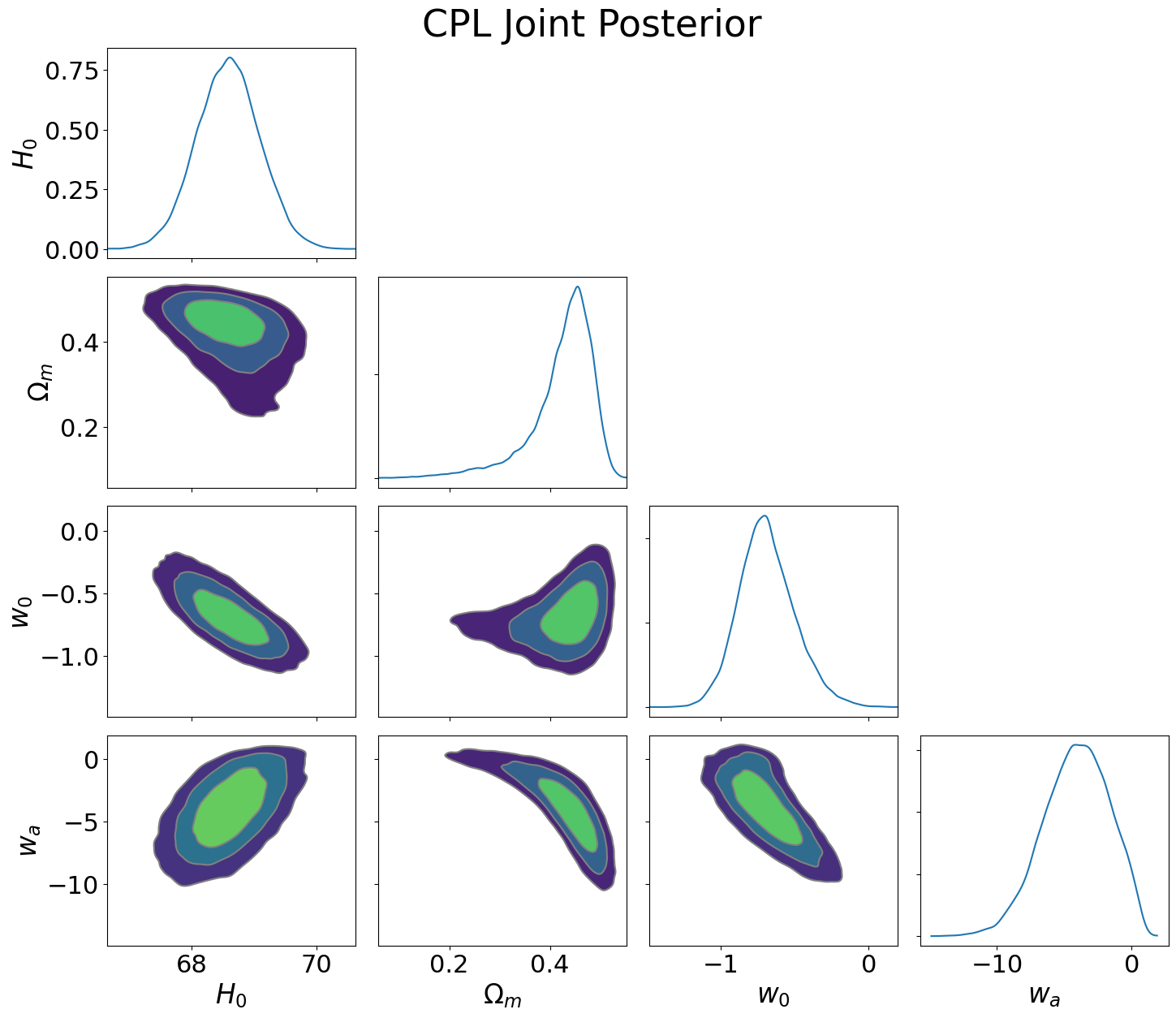}
    \caption{Posterior distributions for cosmological parameters under the \(CPL\) model structured similarly to Fig. 9. The parameters include \(H_0\), \(\Omega_m\), \(w_0\), and \(w_a\). The posteriors show correlated degeneracies among \((w_0, w_a)\), \((w_0, \Omega_m)\), and \((w_a, \Omega_m)\), consistent with the trade-offs between matter density and the evolving dark-energy equation-of-state in \(w(z) = w_0 + w_a \frac{z}{1+z}\).}
    \label{fig:cpl_posterior}
\end{figure}

The joint posterior distributions in \hyperref[fig:lcdm_posterior]{Figure 9}, \hyperref[fig:wcdm_posterior]{Figure 10} and \hyperref[fig:cpl_posterior]{Figure 11}, reveal stretched degeneracies, most prominently in the $(H_0, \Omega_{m})$ for \(\Lambda\)CDM, $(w, \Omega_{m})$ for \(w\)CDM and $(w_{0}, w_{a}), (w_{0}, \Omega_m), (w_{a}, \Omega_m)$ for CPL.  These curved contours arise because Type Ia supernovae constrain the luminosity distance--redshift relation, which depends on integrals of the Hubble parameter. Different parameter combinations can therefore yield nearly indistinguishable expansion histories, resulting in extended,  non-Gaussian degeneracy directions. For example, an increase in $H_{0}$ can be compensated by a reduction in $\Omega_{m}$, while in the CPL parameterisation $w(z) = w_{0} + w_{a} \frac{z}{1+z}$, the parameters $w_{0}$ and $w_{a}$ trade off to maintain similar effective dark energy evolution. Such elongated constraints have been documented in previous supernova analyses \cite{Brout2022, DES2024}.

\begin{table}[H]
    \centering
    \footnotesize
    \begin{tabular}{lrrrrrrr}
    \toprule
    Model & Rank & elpd & \(p_{\mathrm{eff}}\) & elpd difference & Weight & SE & dSE \\
    CPL - DES    & 0 & -2501.24  & 3.39  & 0.00   & 1.00          & 28.28  & 0.0 \\
    \(\Lambda\)CDM - DES   & 1 & -2502.88  & 1.97  & 1.64    & 0  & 28.49  & 1.70 \\
    wCDM - DES & 2 & -2502.90 & 2.59 & 1.66     & 0         & 28.47  & 1.30 \\
    \bottomrule
    \end{tabular}
    \caption{WAIC model comparison summaries on DES data. Columns: rank, expected log predictive density (elpd), effective number of parameters \(p_{\mathrm{eff}}\), elpd difference relative to best model, model weight, standard error (SE), and difference standard error (dSE).}
    \label{tab:waic-comparison-des}
\end{table}

The WAIC results described in \hyperref[sec:Model_Select]{Section 4.6}, summarised in \hyperref[tab:waic-comparison-des]{Table 5}, show that all three models achieve nearly identical predictive performance on the DES supernova dataset. The best-ranked model is CPL with an elpd of $-2501.24$, followed closely by $\Lambda$CDM and $w$CDM with elpd values of \(-2502.88\) and \(-2502.90\), respectively.  The elpd differences between the models are significantly less than the \(SE\) of each elpd, which implies that no meaningful distinction in terms of predictive performance can be made for the models \cite{vehtari2017}.

The model weights, computed from the WAIC elpd values, assign weight \(1.00\) to CPL and zero to both \(\Lambda\)CDM and \(w\)CDM. These weights represent normalised relative predictive support rather than posterior model probabilities, and their extremal values here reflect the exponential scaling of the weighting scheme. Given that the elpd differences are \(<2\) and much smaller than their uncertainties, this outcome should not be interpreted as strong evidence in favour of CPL. Rather, it reinforces the conclusion that all three models are statistically indistinguishable in terms of predictive performance.

In practical terms, this indicates that the DES supernova dataset alone does not provide enough information to distinguish between a cosmological constant, a constant-\(w\) dark energy model, or a time-varying equation of state. WAIC therefore suggests that all three frameworks are equally capable of explaining the observed luminosity–distance relation of Type Ia supernovae. This can be seen visually in \hyperref[fig:hubble-diagram2]{Figure~12} and \hyperref[fig:model-diagnostics]{Figure~13}, where the fitted Hubble diagrams for all three models overlap almost perfectly, even when accounting for their \(68\%\) credible intervals. This illustrates that, consistent with the WAIC comparison, the predictive performance of \(\Lambda\)CDM, \(w\)CDM, and CPL is statistically indistinguishable given the DES-SN5YR data. Furthermore, the residual plots exhibit nearly identical magnitudes and distributions across models, providing additional support for this conclusion of indistinguishability.

\begin{figure}[H]
    \centering

        \includegraphics[width=\linewidth]{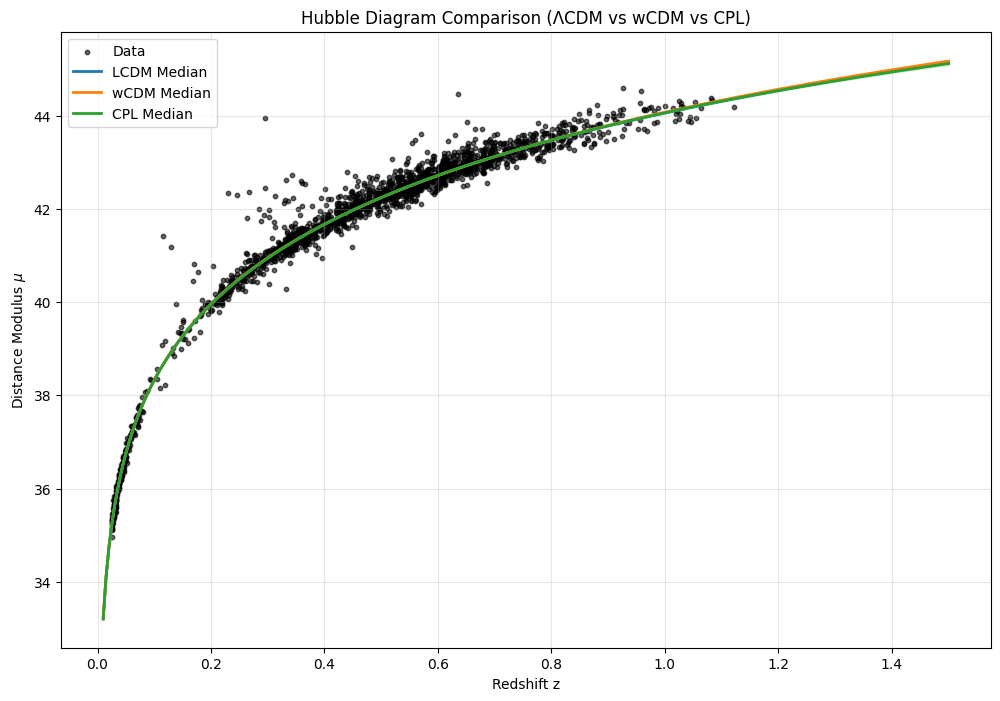}
        
    \caption{Hubble diagram showing the observed Type Ia supernovae from the DES-SN5YR dataset (black points) alongside the theoretical distance–modulus–redshift relations for the \(\Lambda\)CDM, \(w\)CDM, and CPL cosmological models. The solid lines represent the mean theoretical predictions based on the posterior means of each model, while the shaded bands correspond to the \(68\%\) and \(95\%\) credible intervals derived from the posterior predictive distributions. The overlap among the three theoretical curves indicates that, within the current data uncertainties, all models provide comparably good fits to the observed luminosity–distance relation. The credible intervals widen at higher redshifts (though very slightly), reflecting the reduced constraining power of supernova data in that range.}
    \label{fig:hubble-diagram2}
\end{figure}
\begin{figure}[H]
    \centering

        \includegraphics[width=\linewidth]{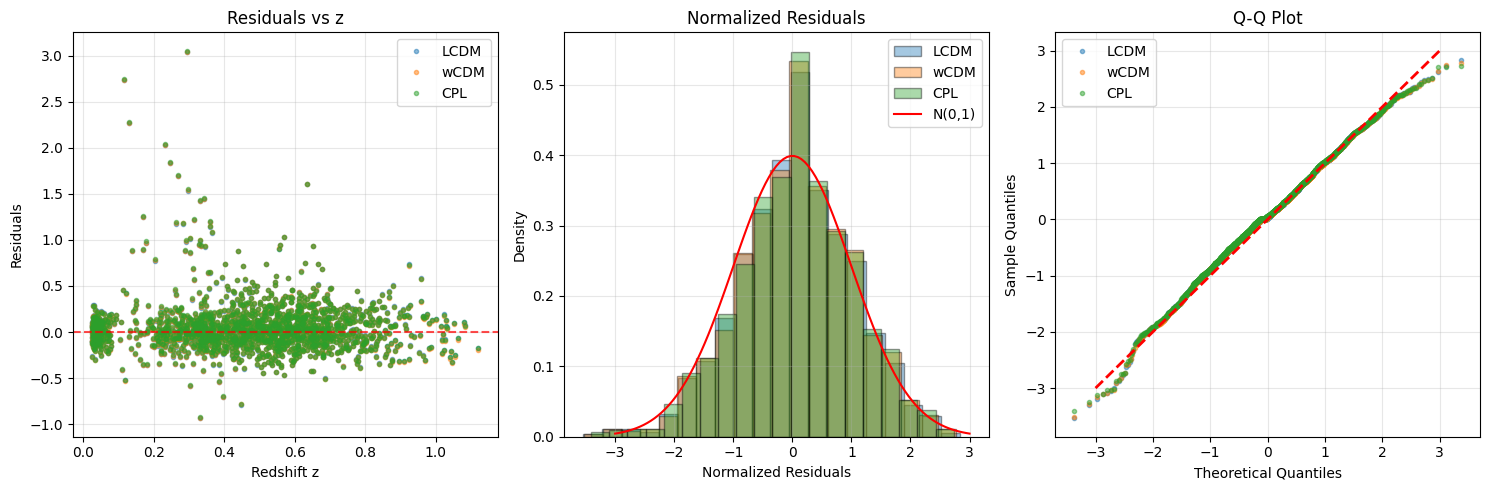}
        
    \caption{Diagnostic plots comparing the residuals for the \(\Lambda\)CDM, \(w\)CDM, and CPL models. The left panel shows the residuals (observed minus model-predicted distance modulus) as a function of redshift, with no systematic trends visible across models. The middle panel displays the distribution of residuals for each model, all approximately centered around zero and following a Gaussian distribution. The right panel shows Q–Q plots comparing the residuals to a standard normal distribution, confirming approximate Gaussianity across models. The diagnostics show no strong evidence of systematic bias.}
    \label{fig:model-diagnostics}
\end{figure}

The Bayes factor results, computed via bridge sampling with 1000 iterations and an estimated \(2\%\) error on the evidence, are summarised in \hyperref[tab:bridge-bayes-summary]{Table~\ref*{tab:bridge-bayes-summary}}. In contrast to the WAIC analysis, which indicated that the three models are statistically indistinguishable in predictive performance, the Bayes factor comparison provides strong and in some cases decisive evidence favouring particular models. 

The comparison between \(\Lambda\)CDM and \(w\)CDM yields a Bayes factor of \(9.70 \times 10^{-29}\) (log Bayes factor \(\approx -64.5\)), representing decisive evidence in favour of \(w\)CDM over \(\Lambda\)CDM. Comparing \(\Lambda\)CDM with CPL gives a Bayes factor of \(4.57 \times 10^{21}\) (log Bayes factor \(\approx 49.9\)), decisively favouring \(\Lambda\)CDM over CPL. Finally, the comparison of \(w\)CDM with CPL produces a Bayes factor of \(4.71 \times 10^{49}\) (log Bayes factor \(\approx 114.4\)), decisively supporting \(w\)CDM over CPL. These values are far beyond the conventional Jeffreys thresholds for ``decisive'' evidence \cite{Robert2007}. 
 
\begin{table}[H]
    \centering
    \resizebox{\textwidth}{!}{
    \begin{tabular}{lp{3.2cm}p{3.2cm}rrr}
    \toprule
    Model 1 & Model 2 & Log Marginal Likelihood \(\log P(D|M_1)\)& Log Marginal Likelihood \(\log P(D|M_2)\) & Bayes Factor \(BF_{1,2}\) & \(\log BF_{1,2}\) \\
    \midrule
    LCDM    & wCDM    & -53557.62 & -53493.11 & \(9.70 \times 10^{-29}\) & -64.51 \\
    LCDM    & CPL     & -53557.62 & -53607.49 & \(4.57 \times 10^{21}\)  & 49.87 \\
    wCDM    & CPL     & -53493.11 & -53607.49 & \(4.71 \times 10^{49}\)  & 114.38 \\
    \bottomrule
    \end{tabular}
    }
    \caption{Bridge sampling log marginal likelihoods, Bayes factors, and log Bayes factors for pairwise cosmological model comparisons. Bayes factors are computed as \(BF_{1,2} = \exp(\log P(D|M_1) - \log P(D|M_2))\).}
    \label{tab:bridge-bayes-summary}
\end{table}

\section{Conclusions and Discussion}

\subsection{Summary of results}

In this paper, we carried out a Bayesian model comparison of three cosmological models: \(\Lambda\)CDM, \(w\)CDM, and CPL using Type Ia supernovae from the DES-SN5YR dataset. We employed two complementary approaches: the Widely Applicable Information Criterion (WAIC) to evaluate predictive accuracy, and Bayes factors via bridge sampling to compare marginal likelihoods.

The WAIC analysis indicated that the differences in expected log predictive density (elpd) between the models are minimal, with all pairwise elpd differences lying below two. According to the thresholds described in \cite{vehtari2017}, such small differences compared to the overall standard error of the elpd imply that the models are statistically indistinguishable in predictive accuracy. None of the models emerge as decisively superior on DES data, suggesting that the information criterion struggled to strongly differentiate the three cosmological models under a dataset that is not highly constraining.

In contrast, Bayes factor analysis provided decisive evidence for model preferences. The results from bridge sampling showed that \(w\)CDM is decisively preferred over \(\Lambda\)CDM (log Bayes factor \(\approx -64.5\)) and even more strongly over CPL (log Bayes factor \(\approx 114\)). Furthermore, \(\Lambda\)CDM was strongly favoured over CPL (log Bayes factor \(\approx 50\)). Taken together, these results establish a clear model hierarchy:
\[
w\mathrm{CDM} \succ \Lambda\mathrm{CDM} \succ \mathrm{CPL},
\]
with overwhelming Bayesian evidence against CPL.

\subsection{Interpretation and implications}

WAIC measures out-of-sample predictive accuracy, while Bayes factors provide a Bayesian measure of overall model plausibility, balancing fit and complexity. The agreement of WAIC across models suggests that, for short-term predictive purposes, \(\Lambda\)CDM, \(w\)CDM, and CPL perform equivalently on the DES-SN5YR dataset. By contrast, Bayes factors show that, when judged by fully Bayesian evidence, \(w\)CDM is strongly preferred.

The \(\Lambda\)CDM model remains the concordant framework, supported by CMB anisotropies \cite{brout2022pantheonplus}, BAO \cite{scolnic2018pantheon}, and large-scale structure surveys. However, our results indicate that, when evaluated on DES supernova data in isolation, extensions allowing a constant but free dark energy equation-of-state parameter \(w\) are favoured. This aligns with other supernova analyses (e.g.~\cite{scolnic2018pantheon, DES2024} ), which report mild but persistent tensions with \(\Lambda\)CDM under the \(w=-1\) assumption.

The CPL parametrisation, designed to capture redshift evolution in \(w\), is decisively disfavoured in our Bayes factor results. This suggests that within the constraining power of DES-SN5YR, the data provide no support for additional model flexibility beyond a constant \(w\). 

\subsection{Limitations}

Several caveats qualify these results. First, our \(H_0\) posterior is conditional on the prior and the choice to sample \(H_0\) explicitly rather than marginalising over the absolute magnitude nuisance parameter. As established in the literature, supernovae alone cannot determine the absolute distance scale without external anchors such as Cepheids, TRGB, or masers \cite{DES2024}. Our \(H_0\) values should therefore not be interpreted as stand-alone supernova measurements but rather as complementary, enabling direct comparison with DES’s marginalised treatment. This caveat does not affect our constraints on \(\Omega_m\), \(w\), \(w_0\), and \(w_a\), which remain consistent with DES and the broader literature.

A further limitation is the redshift range accessible to Type Ia supernovae. The DES-SN5YR data extend only to \(z \sim 1.4\), whereas probes such as BAO and the CMB extend to far higher redshifts (\(z \sim 100-1000\)). As a result, our analysis is confined to the relatively recent expansion history of the Universe and cannot constrain early-universe physics. This restriction also exacerbates degeneracies between cosmological parameters that could otherwise be broken by complementary datasets.

From a computational perspective, we were limited to four HMC chains of 10,000 steps each. Although all convergence diagnostics were satisfactory, including \(\hat{r} \approx 1\) and large effective sample sizes, a more thorough investigation might use longer chains or additional parallel chains (e.g.\ millions of steps) to be more sure of convergence.

Another important caveat concerns prior sensitivity. While our priors were broad and intended to be weakly informative, Bayes factors and, to a lesser extent, parameter posteriors can be sensitive to prior choices, especially in models with additional degrees of freedom \cite{Robert2007} such as CPL’s \(w_a\). A small examination of the priors was conducted through KL divergence but a systematic prior sensitivity study was not performed here, so the Bayes factor results should be interpreted conditionally on the adopted prior assumptions.

Closely related is the issue of evidence estimation. Bridge sampling is regarded as an accurate and efficient method for marginal likelihood evaluation as detailed in \cite{Gronau2017}, and we report estimated errors of \(\sim 2\%\). While sufficient for qualitative model ranking, such errors still imply some uncertainty in the quantitative Bayes factors which could perhaps be improved by alternative methods, such as nested sampling or thermodynamic integration.

Finally, we assumed a spatially flat universe throughout (\(\Omega_k = 0\)) along with other assumptions. While this is strongly supported by independent probes such as the CMB \cite{Planck2018}, it nevertheless represents an assumption that was hard-coded into the models. Allowing \(\Omega_k\) to vary could, in principle, shift posterior constraints or Bayes factors by altering the balance of model complexity versus fit.

\subsection{Future directions}

Next-generation supernova surveys with larger samples and greater redshift range than DES could offer more information to perform more justified inference and comparison between cosmological models. Combining these new supernova datasets with complementary probes such as the CMB, BAO, and weak lensing might then allow degeneracies to be broken and provide a more complete analysis. 

Beyond data, methodological changes might make a difference in the analysis. One such direction is the use of hierarchical Bayesian models as detailed in \cite{Robert2007}. This approach could yield more balanced and rigorous conclusions. 

\subsection{Conclusion}

In conclusion, our analysis finds that the \(\Lambda\)CDM, \(w\)CDM, and CPL models appear nearly indistinguishable in predictive performance on the DES supernova data when judged by WAIC. However, with Bayes factors decisively favouring \(w\)CDM over both \(\Lambda\)CDM and CPL, with \(\Lambda\)CDM itself strongly preferred to CPL. Taken together, the results, under the given data, suggest that allowing a constant but free dark energy equation-of-state parameter offers a better balance of fit and parsimony than either the fixed \(\Lambda\)CDM case or the more flexible CPL parametrisation. 

\section{Acknowledgments}
I would like to thank Dr Siri Chongchitnan and Dr Paul Chleboun for their guidance, patience, and insightful advice throughout this project. I would also like to thank Connor Cassidy for his help in running the Python programs and for his feedback on my code.

\newpage

\bibliographystyle{unsrt}     
\bibliography{references}

\newpage

\begin{appendices}
\section{CPL Hubble Parameter Derivation}
We start from the dark-energy factor in the general expression:
\[
\exp\!\left[ 3 \int_0^z \frac{1+w(z')}{1+z'} \, dz' \right]
= \exp\!\left[3\int_0^z \frac{1 + w_0 + w_a\frac{z'}{1+z'}}{1+z'} \, dz'\right].
\]
Evaluate the integral in two parts:
\[
3(1+w_0)\int_0^z \frac{dz'}{1+z'} \;+\; 3 w_a \int_0^z \frac{z'}{(1+z')^2}\,dz'.
\]
The first integral gives
\[
\int_0^z \frac{dz'}{1+z'} = \ln(1+z).
\]
For the second, let \(u=1+z'\) so \(z'=u-1\), \(du=dz'\):
\[
\int_0^z \frac{z'}{(1+z')^2}\,dz' = \int_{u=1}^{1+z} \frac{u-1}{u^2}\,du
= \int_1^{1+z} \left(\frac{1}{u}-\frac{1}{u^2}\right) du
= \big[\ln u + \tfrac{1}{u}\big]_{1}^{1+z}
= \ln(1+z) + \frac{1}{1+z} - 1.
\]
Combining them:
\[
3(1+w_0)\ln(1+z) \;+\; 3 w_a\Big(\ln(1+z) + \frac{1}{1+z} - 1\Big)
= 3(1+w_0+w_a)\ln(1+z) \;+\; 3 w_a\Big(\frac{1}{1+z}-1\Big).
\]
Noting \(\dfrac{1}{1+z}-1 = -\dfrac{z}{1+z}\), exponentiate:
\[
\exp\!\Big\{3(1+w_0+w_a)\ln(1+z) + 3 w_a\Big(\frac{1}{1+z}-1\Big)\Big\}
= (1+z)^{3(1+w_0+w_a)} \, \exp\!\Big(-\frac{3 w_a z}{1+z}\Big).
\]
Hence the energy component of the CPL Hubble parameter is
\[
E(z) \;=\; \sqrt{ \Omega_m (1+z)^3 + 
\Omega_{\mathrm{\Lambda}} (1+z)^{3(1+w_0+w_a)} \, \exp\!\Big(-\frac{3 w_a z}{1+z}\Big) }.
\]

\section{Code Availability}
The full code used for Bayesian modeling, posterior inference, and cosmological model comparison is available at the following GitHub repository:

\url{https://github.com/Niko-Gigi/Bayesian-Model-Evaluation-for-Cosmology}

\end{appendices}
\end{document}